\documentclass[]{fairmeta}
\PassOptionsToPackage{numbers}{natbib}
\PassOptionsToPackage{sort}{natbib}
\usepackage{color}
\usepackage{multirow}
\usepackage[utf8]{inputenc} 
\usepackage[T1]{fontenc}   
\usepackage{fullpage}
\usepackage[numbers]{natbib}
\usepackage{wrapfig}

\definecolor{redlinkcolor}{rgb}{0.79607843, 0.25098039, 0.25882353}
\definecolor{bluecitecolor}{rgb}{0,0.36,0.69}
\definecolor{newpurple}{HTML}{000099}
\definecolor{newbrown}{HTML}{E55451}

\definecolor{myblue}{HTML}{000099}
\definecolor{myred}{HTML}{E55451}
\definecolor{mycyan}{HTML}{00ccff}
\usepackage{color-edits}

\usepackage{microtype}
\usepackage{graphicx}
\usepackage{booktabs} 

\usepackage{hyperref}

\usepackage{amsmath}
\usepackage{amssymb}
\usepackage{mathtools}
\usepackage{amsthm}

\usepackage{natbib}
\usepackage[utf8]{inputenc} 
\usepackage[T1]{fontenc}    
\usepackage{url}            
\usepackage{booktabs}       
\usepackage{amsfonts}       
\usepackage{nicefrac}       
\usepackage{xcolor}         
\usepackage{color}
\usepackage{algorithm}
\usepackage{algorithmic}
\usepackage{caption}
\usepackage{subcaption}
\usepackage{tikz}
\usetikzlibrary{calc,positioning,tikzmark}
\usepackage{enumitem}
\usepackage[font=small,labelfont=bf]{caption}

\theoremstyle{plain}
\newtheorem{theorem}{Theorem}[section]
\newtheorem{proposition}[theorem]{Proposition}
\newtheorem{lemma}[theorem]{Lemma}

\theoremstyle{definition}
\newtheorem{definition}[theorem]{Definition}

\theoremstyle{remark}

\usepackage[textsize=tiny]{todonotes}

\title{Privacy Amplification for the Gaussian Mechanism via Bounded Support}

\author[1,*]{Shengyuan Hu}
\author[2]{Saeed Mahloujifar}
\author[1]{Virginia Smith}
\author[2]{Kamalika Chaudhuri}
\author[2]{Chuan Guo}

\affiliation[1]{CMU}
\affiliation[2]{FAIR, Meta}
\contribution[*]{Work done at Meta}

\abstract{
Data-dependent privacy accounting frameworks such as per-instance differential privacy (pDP) and Fisher information loss (FIL) confer fine-grained privacy guarantees for individuals in a fixed training dataset. These guarantees can be desirable compared to vanilla DP in real world settings as they tightly upper-bound the privacy leakage for a \emph{specific} individual in an \emph{actual} dataset, rather than considering worst-case datasets. While these frameworks are beginning to gain popularity, to date, there is a lack of private mechanisms that can fully leverage advantages of data-dependent accounting. To bridge this gap, we propose simple modifications of the Gaussian mechanism with bounded support, showing that they {amplify} privacy guarantees under data-dependent accounting. Experiments on model training with DP-SGD show that using bounded-support Gaussian mechanisms can provide a reduction of the pDP bound $\epsilon$ by as much as $30\%$ without negative effects on model utility.}

\date{\today}
\correspondence{Shengyuan Hu at \email{shengyuanhu@cmu.edu}, Chuan Guo at \email{chuanguo@meta.com}}

\begin{document}

\maketitle










\section{Introduction}

Differential privacy (DP) is currently the most common framework for privacy-preserving machine learning~\citep{dwork2006calibrating}. DP upper bounds the privacy leakage of a sample via its privacy parameter $\epsilon$, which holds under worst-case assumptions on the training dataset and randomness in the training algorithm. However, the implied threat model of DP can be too pessimistic in practice, and alternate notions  have emerged to relax stringent worst-case assumptions. In particular, data-dependent privacy accounting frameworks such as per-instance DP (pDP; \citet{wang2019per}) and Fisher information loss (FIL; \citet{hannun2021measuring}) offer more fine-grained privacy assessments for {specific} individuals in an {actual} dataset, rather than considering hypothetical worst-case datasets. This form of privacy guarantee can be more desirable in real world applications as it better captures the capabilities of a realistic adversary.

While data-dependent privacy accounting has been gaining popularity lately~\citep{feldman2021individual, redberg2021privately, pmlr-v162-guo22c, yu2022per, koskela2022individual, boenisch2022individualized}, there are very few private mechanisms that can fully leverage its power. For example, the Gaussian mechanism---arguably the most ubiquitous private mechanism in ML---provides a privacy guarantee that is only dependent on the local sensitivity of the query. Thus, for a simple mean estimation query, all individuals have the same privacy leakage \emph{even under data-dependent accounting} such as pDP (Def. \ref{def:pRDP}) and FIL (Def. \ref{def:fil}).

To bridge this gap, we propose simple modifications of the Gaussian mechanism with bounded support that can \emph{amplify} their privacy guarantee under data-dependent accounting.
One example of such a mechanism is the stochastic sign~\citep{jin2020stochastic}, which first adds centered Gaussian noise $\mathcal{N}(0, \sigma^2)$ and then returns the sign of the output. This mechanism is equivalent to applying the Gaussian mechanism and then performing sign compression, and hence is provably private according to pDP and FIL using post-processing. Intriguingly, a more fine-grained analysis shows that it in fact {amplifies} the privacy guarantee beyond what can be obtained through post-processing. 

Figure \ref{fig:bounded_support_motivation} shows the FIL $\eta$ and per-instance RDP $\epsilon$ (see Def. \ref{def:pRDP}) of the stochastic sign mechanism (orange line) as a function of the input $\theta$. For values of $\theta$ close to 0, the privacy cost is close to that of the standard Gaussian mechanism. However, for values of $\theta$ that are far away from 0, the privacy cost can be drastically lower. Similar observations can be made for the rectified Gaussian and truncated Gaussian mechanisms. We validate these findings empirically on private mean estimation and private model training using DP-SGD~\citep{abadi2016deep}, and show that it is possible to greatly improve the privacy-utility trade-off under our improved analysis.

\begin{figure*}[t!]
    \centering
    \begin{subfigure}{0.24\textwidth}
        \centering
        \includegraphics[width=0.99\textwidth,trim=10 10 10 30]{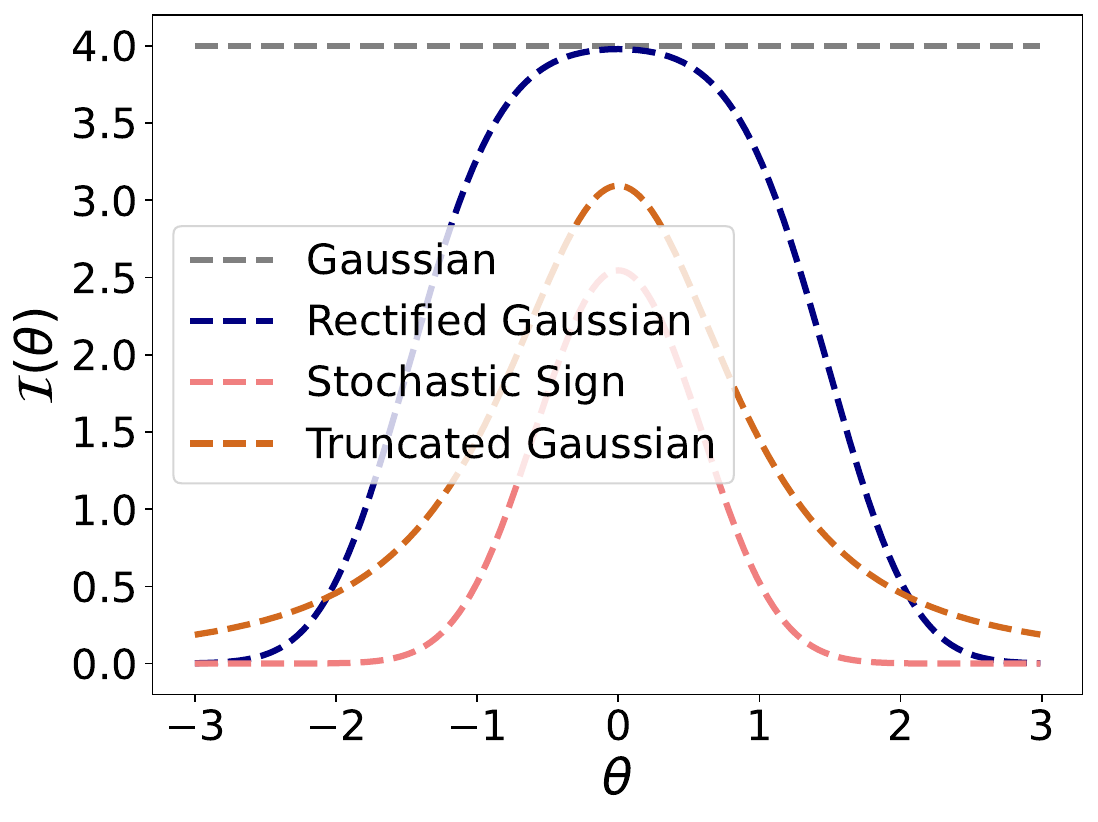}
        \subcaption{FIL,$\sigma=\frac{1}{2}$}
    \end{subfigure}\hfill
    \begin{subfigure}{0.24\textwidth}
        \centering
        \includegraphics[width=0.99\textwidth,trim=10 10 10 30]{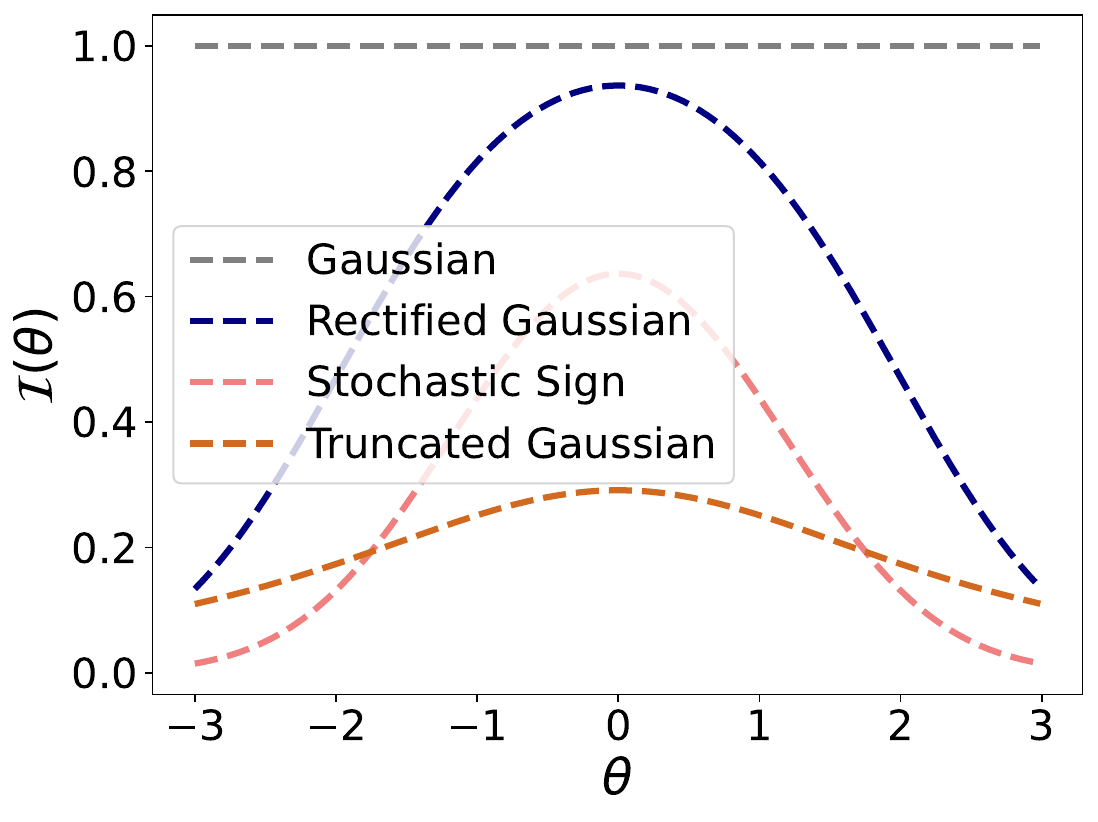}
        \subcaption{FIL,$\sigma=1$}
    \end{subfigure}\hfill
    \begin{subfigure}{0.24\textwidth}
        \centering
        \includegraphics[width=0.99\textwidth,trim=10 10 10 30]{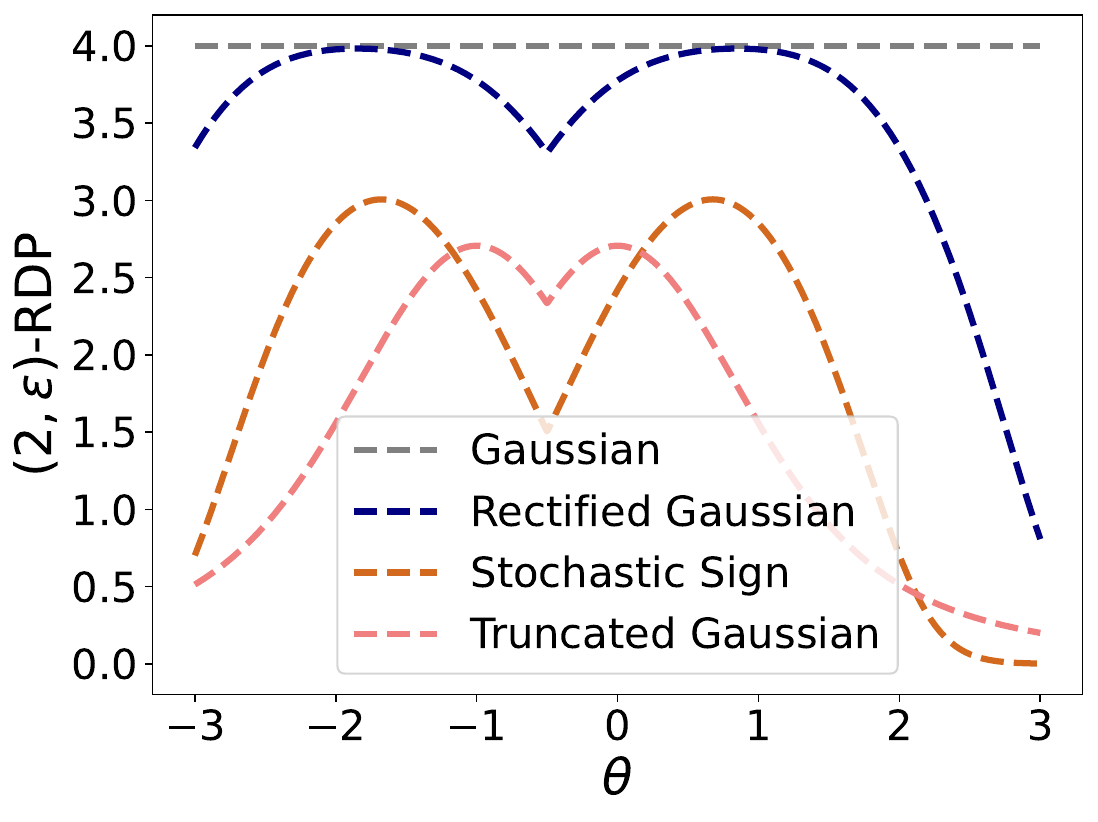}
        \subcaption{\small Per-instance RDP,$\sigma=\frac{1}{2}$}
    \end{subfigure}\hfill
    \begin{subfigure}{0.24\textwidth}
        \centering
        \includegraphics[width=0.99\textwidth,trim=10 10 10 30]{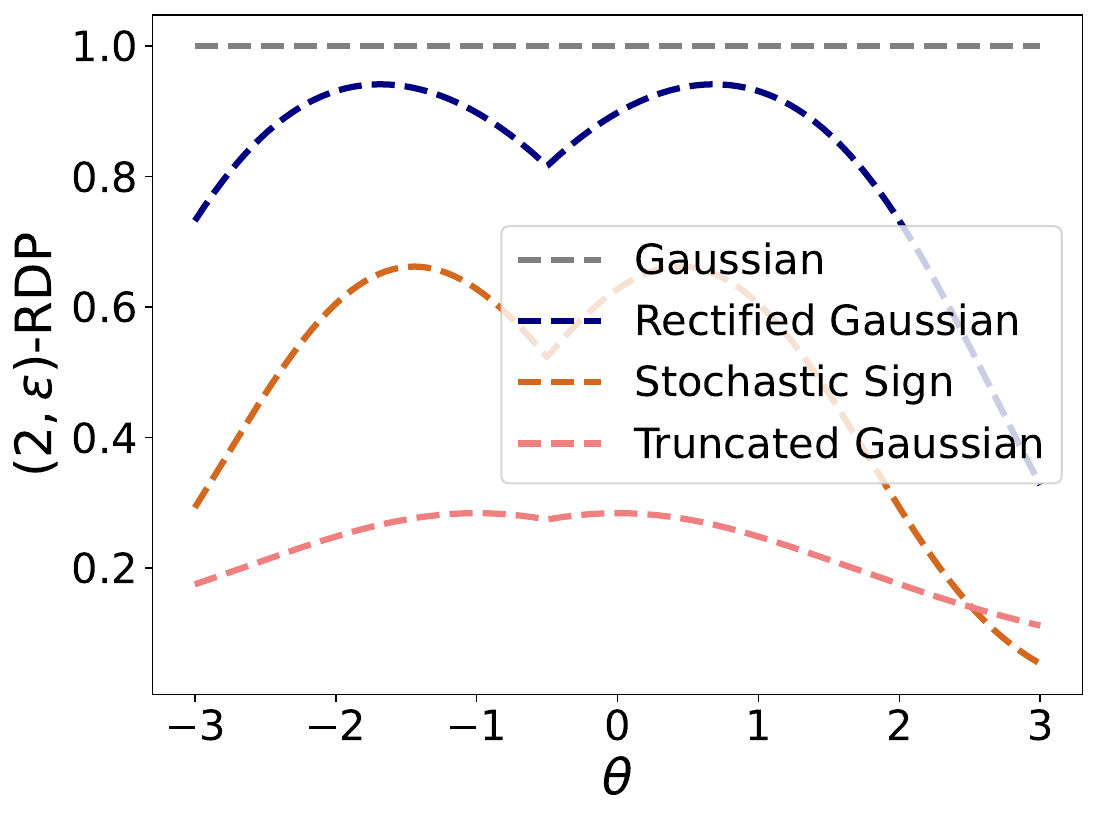}
        \subcaption{\small Per-instance RDP,$\sigma=1$}
    \end{subfigure}
    \caption{\small FIL $\eta$ and per-instance RDP $\epsilon$ for different variants of the Gaussian mechanism with input $\theta$. Compared to vanilla Gaussian mechanism, rectified Gaussian, truncated Gaussian and stochastic sign enjoy stronger privacy guarantee, especially when $|\theta|$ is large. Bounded support set $\mathcal{B}=[-1,1]$ for all figures. Sensitivity is controlled at 1 for RDP results.} %
    \label{fig:bounded_support_motivation}
\end{figure*}

\paragraph{Contributions.} We list our main contributions below:
\begin{enumerate}[leftmargin=*]
\itemsep0em 
    \item We analyze two modified Gaussian mechanisms with bounded support---the \textit{rectified Gaussian mechanism} and \textit{truncated Gaussian mechanism}---under the data-dependent privacy frameworks of per-instance DP (pDP) and Fisher information loss (FIL).
    \item Our  analysis shows that bounded-support Gaussian mechanisms amplify the privacy guarantee for both privacy metrics compared to the vanilla Gaussian mechanism under the same noise scale $\sigma$. The amount of amplification depends on the true location parameter $\theta$, \emph{i.e.}, the input. 
    \item Empirically, we demonstrate that our analysis of the bounded-support Gaussian mechanisms provides improved privacy-utility trade-off in terms of both pDP and FIL on several real world image classification tasks. For example, on CIFAR-100, bounded-support Gaussian mechanisms reduce the pDP bound $\epsilon$ by as much as $>30\%$ with equal test accuracy.

\end{enumerate}

\section{Background \& Related Work}
We first introduce background on privacy tools used in this work, in particular on differential privacy and DP-SGD (Section \ref{background:dp}), and Fisher information loss (Section \ref{background:fil}).

\subsection{Differential Privacy (DP)}
\label{background:dp}
We start with the traditional definition of differential privacy (DP) \citep{DworkR14}.
\begin{definition}[Differential Privacy; \citet{DworkR14}] A randomized algorithm $\mathcal{M}:\mathcal{X}\rightarrow\mathcal{R}$ satisfies $(\epsilon,\delta)$-DP if for every neighboring pairs of datasets $D,D'\sim\mathcal{X}$ and measurable sets $S\subset\mathcal{R}$, we have \[\text{Pr}[\mathcal{M}(D)\in S]\leq e^\epsilon\text{Pr}[\mathcal{M}(D')\in S]+\delta\]

\end{definition}
DP guarantees that the presence of any single data point has little impact on the output distribution of the randomized mechanism. On the other hand, the privacy parameter $\epsilon$ could often be uninformative of the privacy loss of the actual input dataset. Since $(\epsilon,\delta)$-DP is data independent, it fails to fully exploit the power of mechanisms whose privacy loss is not location invariant, e.g. stochastic sign. Hence, in this work, we focus on the following relaxation of DP known as Per-instance Differential Privacy (pDP) for all~\citep{wang2019per} to characterize the privacy guarantee.

\begin{definition}[Per-instance Differential Privacy (pDP) for all; \citet{wang2019per}]
    Given fixed dataset $D$ and any single sample $z\sim\mathcal{X}$, a randomized mechanism $\mathcal{M}:\mathcal{X}\rightarrow\mathcal{R}$ satisfies $(\epsilon,\delta)$\textit{-pDP for all} for $D$ if for every measuarable set $S\subset\mathcal{R}$, we have
    \begin{align*}&\text{Pr}[\mathcal{M}(D)\in S]\leq e^\epsilon\text{Pr}[\mathcal{M}([D,z])\in S]+\delta,\\&\text{Pr}[\mathcal{M}([D,z])\in S]\leq e^\epsilon\text{Pr}[\mathcal{M}(D)\in S]+\delta\end{align*}
\end{definition}
In the remainder of the paper, we use \textit{pDP} as an abbreviation of the above notion. The definition of pDP is closely related to local sensitivity~\citep{nissim2007smooth} where we fix one dataset upfront and only take the max over its neighboring dataset. An important property of pDP is that it allows privacy loss $\epsilon$ to be dependent on the actual dataset $D$. 

In this work we focus on a variant called R\'enyi Differential Privacy (RDP; \citet{mironov2017renyi}), which allows for convenient and tight privacy accounting under composition. Similar to pDP, we extend the RDP definition to allow data dependent account defined below.

\begin{definition}[Per-instance R\'enyi Differential Privacy for all]
\label{def:pRDP}
For a given $\alpha>1$ and dataset $D$, a randomized mechanism $\mathcal{M}:\mathcal{X}\rightarrow\mathcal{R}$ satisfies per-instance $(\alpha,\epsilon)$-RDP for all if for any single sample $z\sim\mathcal{X}$, we have $D_{\alpha}(\mathcal{M}(D)\|\mathcal{M}([D,z]))\leq\epsilon$ where $D_{\alpha}(P\|Q)$ is the R\'enyi Divergence between probability distribution $P$ and $Q$:\[D_{\alpha}(P\|Q)=\frac{1}{\alpha-1}\log\left(\mathbb{E}_{x\sim Q}\left[\left(\frac{P(x)}{Q(x)}\right)^\alpha\right]\right).\]
    
\end{definition}




\textbf{DP-SGD.} To apply differential privacy to iterative gradient-based learning algorithms such as SGD, \citet{abadi2016deep} proposed DP-SGD based on the Gaussian mechanism. Intuitively, at each iteration, DP-SGD performs $\ell_2$ clipping for per-example gradient before aggregation and applies the Gaussian mechanism to the aggregated gradient. 
By applying composition and post-processing, one can show that DP-SGD provides provable $(\epsilon,\delta)$-DP guarantee.
Unfortunately, DP-SGD is known to suffer from undesirable privacy-utility trade-offs \citep{bagdasaryan2019differential}, \emph{i.e.}, utility degrades as the privacy budget $\epsilon$ decreases.

\subsection{Fisher Information Loss (FIL)}
\label{background:fil}

The idea of DP is centered around hypothesis testing, asserting that an adversary cannot distinguish between the output of a private mechanism $\mathcal{M}$ when executed on adjacent datasets.
\citet{hannun2021measuring} proposed an alternative privacy notion known as Fisher information loss (FIL) that instead opts for the parameter estimation interpretation.
\begin{definition}[Fisher information loss; \citet{hannun2021measuring}]
    \label{def:fil}
    Given a dataset $D\sim\mathcal{X}$, we say that a randomized algorithm $\mathcal{M}:\mathcal{X}\rightarrow\mathcal{R}$ satisfies FIL of $\eta$ w.r.t. $D$ if $\|\mathcal{I}_h(D)\|_2\leq\eta$, where \[\mathcal{I}_h(D)=-\mathbb{E}_{h}\left[\nabla^2_D\log\text{Pr}_{\mathcal{M}}(h|D)\right]\]
    is the Fisher information matrix (FIM), $\|\cdot\|_2$ denotes the matrix 2-norm, and $\nabla^2_D$ denotes the Hessian matrix w.r.t $D$.
\end{definition}
The privacy implication of FIL is given by the Cram\'{e}r-Rao bound, which states that any unbiased estimate of the dataset $D$ has variance lower bounded by $\frac{1}{\eta^2}$ \citep{hannun2021measuring}. In other words, a smaller FIL implies the dataset $D$ is harder to estimate given output $\mathcal{M}(D)$. FIL can also be computed for arbitrary subsets of $D$ such as sample-wise and group-wise. Furthermore, \citet{pmlr-v162-guo22c} derived FIL accounting for DP-SGD by proving the composition theorem and amplification via subsampling for FIL. 

\subsection{Compression-aware Privacy Mechanism}
Several prior efforts tried to study bounded noise mechanism such as bounded Laplace \citep{holohan2018bounded}, generalized gaussian \citep{liu2018generalized} in tasks such as adaptive data analysis and query answering \citep{dagan2022bounded}. However, these work don't show privacy amplification via bounded support, and they do not focus on private SGD which is the main application of our new mechanism.
Our work is also closely related to compression-aware privacy mechanisms~\citep{canonne2020discrete, agarwal2021skellam,chen2022poisson,guo2022interpolated}, where the mechanism's output is compressed to a small number of bits for efficient communication without significantly harming the privacy-utility trade off. Similarly, these mechanisms do not provide amplified privacy themselves. Our work leverages ideas from the compression literature to design mechanisms that come with \textit{improved} privacy guarantees. 


\section{Privacy Mechanisms}

Private gradient-based optimization can be surprisingly robust to compression. One intriguing example is stochastic sign~\citep{jin2020stochastic}, which applies sign compression on top of the Gaussian mechanism. On one hand, taking the sign reduces the sparse noisy signal in the gradient so that the optimization could is not dominated by those noise. Hence, sign SGD can achieve comparable or even better performance compared to SGD~\citep{bernstein2018signsgd}. Meanwhile, as we only have one bit information per coordinate under stochastic sign, we should get amplified privacy compared to vanilla Gaussian mechanism by post-processing theorem for DP. Unfortunately, such amplification has not been investigated empirically and theoretically previously.

In this work, we propose two alternatives to the Gaussian mechanism, both of which utilize the idea of having a probability density function with pre-determined bounded support. Throughout the rest of the paper, we will use $\phi(x)$ (resp. $\Phi(x)$) to represent the standard Gaussian pdf (resp. cdf). Our first example is to clip the output of Gaussian perturbation. Such clipping leverages similar intuition to reduce noisy signal in the gradient as stochastic sign.
We introduce our first mechanism based on \textit{rectified Gaussian distribution} defined as the following.

\begin{definition}
    \label{eq:rectified_gaussian}
    A random variable $X$ follows a \textit{rectified Gaussian distribution\footnote{Also referred as boundary inflated Gaussian distribution in other works \citep{liu2018generalized}.}}, denoted as $X\sim\mathcal{N}^R(\mu,\sigma^2,[a,b])$ for some $\mu,\sigma>0,a<b$ if its probability density function is the following:
    \begin{equation}
    \begin{split}
        p(x) =\left(\frac{1}{\sigma}\phi\left(\frac{x-\mu}{\sigma}\right)\right)^{U(x;a,b)}\times\left(\Phi\left(\frac{a-\mu}{\sigma}\right)\delta(x;a)\right)^{\mathbf{1}_{x=a}}\left(\Phi\left(\frac{\mu-b}{\sigma}\right)\delta(x;b)\right)^{\mathbf{1}_{x=b}}
    \end{split}
    \end{equation}
    where $U(x;a,b)=1$ iff $a<x<b$ and 0 otherwise; $\mathbf{1}_{x=i}=1$ iff $x=i$ and 0 otherwise; $\delta(x;i)=+\infty$ iff $x=i$ and 0 otherwise.
    
\end{definition}
Note that the rectified Gaussian is a mixture of a continuous random variable and discrete random variable. When $x\in(a,b)$ it follows the distribution of a Gaussian random variable. Otherwise, $x$ either equals to $a$ with probability $\Phi\left(\frac{a-\mu}{\sigma}\right)$ or equals to $b$ with probability $\Phi\left(\frac{\mu-b}{\sigma}\right)$. Therefore, its support is a bounded interval $[a,b]$. Sampling from a rectified Gaussian is equivalent to first sampling from a Gaussian and then clip the value onto bounded interval $[a,b]$. Therefore, a rectified Gaussian random variable contains strictly equivalent or less information about the true location compared to a Gaussian random variable.

Closely related to the rectified Gaussian, the \textit{truncated Gaussian distribution} enjoys the same property of having bounded support while handling the density at the tail differently. The definition is given as the following:

\begin{definition}
    A random variable $X$ follows a \textit{truncated Gaussian distribution}, denoted as $X\sim\mathcal{N}^T(\mu,\sigma^2,[a,b])$ for some $\mu,\sigma>0,a<b$ if its probability density function is the following:
    \begin{equation}
        p(x)=\frac{1}{\sigma}\frac{\phi\left(\frac{x-\mu}{\sigma}\right)}{\Phi\left(\frac{b-\mu}{\sigma}\right)-\Phi\left(\frac{a-\mu}{\sigma}\right)}U(x;a,b)
    \end{equation}
\end{definition}

Unlike rectified Gaussian where the Gaussian probability density at the tail is concentrated at the two ends of the closed support set, the truncated Gaussian directly normalizes the Gaussian probability density within the support set. It is also not straightforward to argue the truncated Gaussian provides amplified privacy since it is not achievable from post processing a Gaussian. 

In the remainder of the paper, we will use $\mathcal{N}^{B}$ as a generalized expression for bounded Gaussian noise, e.g. $\mathcal{N}^R$ and $\mathcal{N}^T$. Similar to the Gaussian mechanism, where we sample a noisy output hypothesis from a Gaussian distribution, we refer the process of sampling from $\mathcal{N}^{B}$ as the bounded Gaussian mechanism defined below.

\begin{definition}[Bounded Gaussian mechanism]
    Let $D\in\mathcal{D}$ be the dataset and $f:\mathcal{D}\rightarrow\mathbb{R}^d$ be a real function. We define the Bounded Gaussian mechanism with support set $\mathcal{B}$ as \begin{equation*}
        BG(f(D))\sim\mathcal{N}^B\left(f(D),\sigma^2\mathbf{I}_d, \mathcal{B}\right).
    \end{equation*}
\end{definition}

Note that our mechanism differs from additive noise mechanism where one adds an independent, bounded noise term to perturb the mechanism input. Instead, we directly sample from a bounded noise distribution with the mechanism input as the location parameter and a data independent bounded support set. The former will not have additional privacy amplification at the tail as the support set for the mechanism changes as the input changes.

\section{Privacy Amplification}

In this section we show how bounded Gaussian mechanism enjoys amplified FIL (Section \ref{subsec:fil}) and pDP (Section \ref{subsec:rdp}) compared to Gaussian mechanism.
\subsection{Privacy Amplification for FIL}
\label{subsec:fil}
\textbf{Setup.} Consider the setting where $D={(x_i,y_i)}_{i=1,\cdots,n}$ is the training set. Let $f$ be a deterministic function that maps $D$ to some $d$ dimensional real vector. 

\textbf{Compute FIL.} Similar to \citet{hannun2021measuring}, we define $J_f$ to be the Jacobian matrix of $f(D)$ with respect to $D$. We can calculate the closed form FIM for both the rectified Gaussian mechanism and truncated Gaussian mechanism. We will use $\theta=f(D)$ for the rest of the section.
\begin{lemma}
\label{lemma:fil_closed_form}
    The FIL of $\theta^T\sim\mathcal{N}^T\left(\theta,\sigma^2,[a,b]\right)$ is given by $\eta=\eta^T\|J_f\|_2$ where
    \begin{equation}
        \begin{split}
            (\eta^T)^2 = \frac{1}{\sigma^2}-\frac{1}{\sigma^2}\frac{\left(\phi\left(\frac{b-\theta}{\sigma}\right)-\phi\left(\frac{a-\theta}{\sigma}\right)\right)^2}{\left(\Phi\left(\frac{b-\theta}{\sigma}\right)-\Phi\left(\frac{a-\theta}{\sigma}\right)\right)^2}+\frac{1}{\sigma^2}\frac{\frac{a-\theta}{\sigma}\phi\left(\frac{a-\theta}{\sigma}\right)-\frac{b-\theta}{\sigma}\phi\left(\frac{b-\theta}{\sigma}\right)}{\Phi\left(\frac{b-\theta}{\sigma}\right)-\Phi\left(\frac{a-\theta}{\sigma}\right)}.
        \end{split}
    \end{equation}
    The FIL of $\theta^R\sim\mathcal{N}^R\left(\theta,\sigma^2,[a,b]\right)$ is given by $\eta=\eta^R\|J_f\|_2$ where
    \begin{equation}
    \label{eq:rec_fil}
        \begin{split}
            (\eta^R)^2 &= \frac{1}{\sigma^2}\left(\frac{\phi^2\left(\frac{a-\theta}{\sigma}\right)}{\Phi\left(\frac{a-\theta}{\sigma}\right)}+\frac{\phi^2\left(\frac{\theta-b}{\sigma}\right)}{\Phi\left(\frac{\theta-b}{\sigma}\right)}\right)+\frac{1}{\sigma^2}\left(\Phi\left(\frac{\theta-b}{\sigma}\right) - \Phi\left(\frac{\theta-a}{\sigma}\right)\right)\\ &+\frac{1}{\sigma^2}\left(\frac{a-\theta}{\sigma}\phi\left(\frac{a-\theta}{\sigma}\right) - \frac{b-\theta}{\sigma}\phi\left(\frac{b-\theta}{\sigma}\right)\right).
        \end{split}
    \end{equation}
\end{lemma}

We provide a detailed derivation of the above Lemma in the Appendix \ref{appen:fil}. Prior work has shown that the FIL of the Gaussian mechanism is given by $\eta^G\|J_f\|_2 = \frac{1}{\sigma}\|J_f\|_2$ \citep{hannun2021measuring}. Observe that $\eta^G$ is fixed given the variance, which is not the case for $\eta^T$ and $\eta^R$. Hence, when $a,b$ are fixed constant, FIL for bounded Gaussian mechanism changes as the true location parameter changes. For example, consider the rectified Gaussian mechanism and focus on the case where $|\theta|\rightarrow\infty$. It is easy to verify that the last two terms in Equation \ref{eq:rec_fil} converges to 0. Applying L'Hôpital's rule to the first term also gives us the fact that it converges to 0. Hence, $\eta^R$ is asymptotically close to 0 when $|\theta|\rightarrow\infty$, a huge amplification compared to $\eta^G$ which is a constant. In fact, such privacy amplification does not only happen to the tail. We showed that such amplification is general for both mechanisms and all location parameters.

\begin{theorem}
\label{th:bgm_fil_amp}
    Assume $h\sim\mathcal{N}(f(D),\sigma^2)$ has FIL of $\eta$, $h\sim\mathcal{N}^{B}(f(D),\sigma^2,[a,b])$ has FIL of $\eta^B$,. Then for any $f(D)$, we have $\eta^B\leq\eta$.
    This is true for both $\mathcal{N}^R$ and $\mathcal{N}^T$.
\end{theorem}

We defer the proof to Appendix \ref{appen:fil}, and in Figure \ref{fig:bounded_support_motivation}, 
 plot how the FIL differs among different mechanisms with respect to the location parameter $\theta$.  The gap between FIL of Gaussian and FIL of bounded Gaussian is the smallest at $\theta=0$ and gradually grows as $\theta$ approaches the tail. Meanwhile, we also observe that more relaxed bounded support (smaller $\sigma$) results in weaker privacy amplification. 



As mentioned in the introduction, another example of the Gaussian mechanism with bounded support is stochastic sign \citep{jin2020stochastic}, where we take the sign of the Gaussian perturbed gradient as the new gradient value. This is equivalent to 1-bit quantization with bounded range $[-1,1]$. We can compute the FIL of stochastic sign: $\eta^{Sgn}=\frac{\phi\left(\frac{\theta}{\sigma}\right)}{\sigma\sqrt{\Phi\left(\frac{\theta}{\sigma}\right)\Phi\left(\frac{-\theta}{\sigma}\right)}}\|J_f\|_2$. Similar to the rectified Gaussian mechanism, it relies on the location parameter, which decreases as $\theta$ approaches the tail. Comparison between stochastic sign and Gaussian is also shown in Figure \ref{fig:bounded_support_motivation}. 

\subsection{Privacy Amplification for per-instance RDP}
\label{subsec:rdp}
We start by introducing RDP accounting for Gaussian mechanism in the scalar case without subsampling.
\begin{lemma}[Proposition 7 from \citet{mironov2017renyi}]
\label{lemma:gaussian_rdp}
    $D_{\alpha}\left(\mathcal{N}(\theta,\sigma^2)\|\mathcal{N}(\theta+c,\sigma^2)\right)=\frac{\alpha c^2}{2\sigma^2}$
\end{lemma}
Note that we can directly use Lemma \ref{lemma:gaussian_rdp} to show per-instance RDP for all guarantee for Gaussian mechanism because given fixed $\theta$, $c$ maximizes the R\'enyi Divergence between the two Gaussian. 
Unfortunately, it is unclear whether this true for bounded Gaussian mechanism. Therefore, we ask the following question: \textit{Given fixed $\theta$, what is the dominating pair of distribution under R\'enyi Divergence for rectified and truncated Gaussian distribution?} For the rest of the section, we will use $\mathcal{B}=[-a,a], a>0$ as the default bounded support set. 

\subsubsection{Accounting for Rectified Gaussian}
We start from the simpler case of rectified Gaussian.
\textbf{Calculation of RDP.}
 We first present the closed form for the R\'enyi Divergence between two rectified Gaussian distribution with same variance and bounded support.
\begin{lemma} 
\label{lemma:rgm_rdp}
Let $\Delta(x)=\Phi\left(\frac{a-x}{\sigma}\right)-\Phi\left(\frac{-a-x}{\sigma}\right)$
    \begin{equation}
    \label{eq:rec_rdp}
        \begin{split}
            D_{\alpha}(\mathcal{N}^R(\theta,\sigma,\mathcal{B})\|\mathcal{N}^R(\theta+c,\sigma,\mathcal{B}))
    =&\frac{1}{\alpha-1}\log\Bigg(\exp\left(\frac{(\alpha^2-\alpha)c^2}{2\sigma^2}\right)\Delta\left(\theta+(1-\alpha)c\right)\\
    &+\Phi\left(\frac{-a-\theta}{\sigma}\right)^\alpha\Phi\left(\frac{-a-\theta-c}{\sigma}\right)^{1-\alpha}
    +\Phi\left(\frac{\theta-a}{\sigma}\right)^\alpha\Phi\left(\frac{\theta+c-a}{\sigma}\right)^{1-\alpha}\Bigg)
        \end{split}
    \end{equation}
\end{lemma}
\textbf{Privacy Amplification.} As we mentioned earlier, rectified Gaussian mechanism could be viewed as post-processing (by doing clipping) of Gaussian mechanism, we can rely on Data Processing Inequality (DPI) of R\'enyi Divergence.
\begin{lemma}[Theorem 9 from \citet{van2014renyi}]
\label{lemma:rdp_dpi}
    If we fix the transition probability $A_{Y|X}$ for a Markov Chain $X\rightarrow Y$, we have $D_{\alpha}(P_Y\|Q_Y)\leq D_{\alpha}(P_X\|Q_X)$.
\end{lemma}
A direct consequence of Lemma \ref{lemma:rdp_dpi} is that we get privacy amplification though rectification:
\begin{proposition} Given fixed $\theta$ and norm clipping bound $c$, we have $D_{\alpha}\left(\mathcal{N}^R\left(\theta,\sigma^2,\mathcal{B}\right)\|\mathcal{N}^R\left(\theta+c,\sigma^2,\mathcal{B}\right)\right) \leq D_{\alpha}\left(\mathcal{N}\left(\theta,\sigma^2\right)\|\mathcal{N}\left(\theta+c,\sigma^2\right)\right)$
\end{proposition}

\subsubsection{Accounting for Truncated Gaussian}
\textbf{Calculation of RDP.} Similar to the previous section, we
first present how to calculate R\'enyi Divergence for truncated
Gaussian distribution.

\begin{lemma} Let $\Delta$ defined similarly in Lemma \ref{lemma:rgm_rdp}
\label{lemma:tgm_rdp}
    \begin{equation}
    \label{eq:trunc_rdp}
        \begin{split}
            D_{\alpha}(\mathcal{N}^T(\theta,\sigma^2,\mathcal{B})\|\mathcal{N}^T(\theta+c,\sigma^2,\mathcal{B}))
    =D_{\alpha}(\mathcal{N}(\theta,\sigma^2)\|\mathcal{N}(\theta+c,\sigma^2))
    +\log\left(\frac{\Delta(\theta+c)}{\Delta(\theta)}\left(\frac{\Delta((1-\alpha)(c+\theta))}{\Delta(\theta)}\right)^{1-\alpha}\right)
        \end{split}
    \end{equation}
\end{lemma}

\textbf{Privacy Amplification.} As mentioned earlier, since truncated Gaussian is not achievable from post-processing of Gaussian, Lemma \ref{lemma:rdp_dpi} is not directly applicable here. Therefore, following Definition \ref{def:pRDP}, we need to find a pair of distribution that maximizes the R\'enyi Divergence given location parameter $\theta$ and clipping bound $c$, i.e. find \[\max_{x:0<x\leq c}D_{\alpha}\left(\mathcal{N}^R\left(\theta,\sigma^2,S\right)\|\mathcal{N}^R\left(\theta+x,\sigma^2,S\right)\right).\] We can show the monotonicity of R\'enyi Divergence between two truncated Gaussian with respect to sensitivity.

\begin{lemma}
    \label{lemma:monotonicity}
    For any fixed $\theta,\sigma>0$ and $\mathcal{B}=[-a,a]$, we have $D_{\alpha}\left(\mathcal{N}^T\left(\theta,\sigma^2,\mathcal{B}\right)\|\mathcal{N}^T\left(\theta+x,\sigma^2,\mathcal{B}\right)\right)$ is an increasing function of $x$.
\end{lemma}

Lemma \ref{lemma:monotonicity} tells us that similar to the Gaussian case, the R\'enyi Divergence between two bounded Gaussian distribution with same variance is maximized at max sensitivity. Now, given the same sensitivity $c$, we show that regardless of location parameter chosen, we always get smaller R\'enyi Divergence for truncated Gaussian compared to Gaussian. The result is presented below.

\begin{theorem}
\label{th:trunc_dp_amplification}
    Given fixed $\theta$ and norm clipping bound $c$, we have $D_{\alpha}\left(\mathcal{N}^B\left(\theta,\sigma^2,\mathcal{B}\right)\|\mathcal{N}^B\left(\theta+c,\sigma^2,\mathcal{B}\right)\right) \leq D_{\alpha}\left(\mathcal{N}\left(\theta,\sigma^2\right)\|\mathcal{N}\left(\theta+c,\sigma^2\right)\right)$
\end{theorem}
We defer the detailed proof of Theorem \ref{th:trunc_dp_amplification} to Appendix \ref{appen:dp_amp}.




Up till now, for both rectified and truncated Gaussian mechanisms we are able to show that they could achieve smaller R\'enyi Divergence compared to vanilla Gaussian. To perform accounting, it suffices to calculate \begin{align*}
    \max\Bigg\{D_{\alpha}\left(\mathcal{N}^{B}\left(\theta,\sigma^2,\mathcal{B}\right)\|\mathcal{N}^B\left(\theta+c,\sigma^2,\mathcal{B}\right)\right), D_{\alpha}\left(\mathcal{N}^{B}\left(\theta+c,\sigma^2,\mathcal{B}\right)\|\mathcal{N}^B\left(\theta,\sigma^2,\mathcal{B}\right)\right)\Bigg\}
\end{align*}
We show comparison of the $(2,\epsilon)$-RDP for different mechanisms in Figure \ref{fig:bounded_support_motivation}. Similar to FIL, privacy amplification via bounded support is more significant as $\theta$ approaches the tail, i.e. when the absolute value of $\theta$ is large. This is a direct result of having a data independent $S$: Consider the case where $\theta\gg 1$. Let $\hat{\theta}=BG(\theta)$ be the noisy estimate for $\theta$. When $BG$ is the rectified Gaussian mechanism, $\hat{\theta}=1$ with probability $\Phi\left(\frac{\theta-1}{\sigma}\right)\approx 1$ for large enough $\theta$. When $BG$ is the truncated Gaussian mechanism, $\hat{\theta}$ has approximately the same probability of being any value between -1 and 1. Neither case gives us much useful information of the value $\theta$ itself. 
Therefore, bounded Gaussian mechanism provides much stronger privacy protection given this specific $\theta$. 

\setlength{\textfloatsep}{20pt}
\begin{algorithm}[t]
\setlength{\abovedisplayskip}{0pt}
\setlength{\belowdisplayskip}{0pt}
\setlength{\abovedisplayshortskip}{0pt}
\setlength{\belowdisplayshortskip}{0pt}
\caption{RDP \& FIL Accounting for private SGD using bounded Gaussian mechanism}
\label{alg:2}
\begin{algorithmic}[1]
    \STATE {\bf Input:}  Per-example $L_{\infty}$-clipped gradients $\{\hat{g}(x_1),\cdots,\hat{g}(x_n)\}\in\left[\mathbb{R}^d\right]^n$, variance $\sigma^2\mathbf{I}_d$, sensitivity $C$, bounded support set $\mathcal{B}=\{\mathbf{x}|\|\mathbf{x}\|_{\infty}\leq a\}$, RDP parameter $\alpha$.
    \STATE Initialize $\boldsymbol{\epsilon}=\textbf{0}_d, \boldsymbol{\eta}=\textbf{0}_{d}$. Denote $S=[-a,a]$.
    \STATE Compute the summed gradient: $\hat{g}\leftarrow\sum_i\hat{g}(x_i)$.
    \FOR {$j\in [d]$}
        \STATE Compute Renyi Divergence based on Equation \ref{eq:rec_rdp} / \ref{eq:trunc_rdp}:\[\boldsymbol{\epsilon}_j\leftarrow D_{\alpha}\left(\mathcal{N}^B\left(\hat{g}[j],\sigma^2, S\right)\|\mathcal{N}^B\left(\hat{g}[j]+C,\sigma^2, S\right)\right)\]
        \STATE Compute $\eta^B$ based on Lemma \ref{lemma:fil_closed_form} for each example:\[\boldsymbol{\eta}_{j}\leftarrow\eta^B(\hat{g}[j],\sigma^2,S)\]
    \ENDFOR
    \STATE \textbf{Output} RDP budget $\epsilon=\sum_{j}\boldsymbol{\epsilon}_j$ and per example FIL 
    $\textbf{FIL}[i]=\left\|\boldsymbol{\eta}^\top\nabla_{x_i}\hat{g}(x_i)\right\|_2^2$.
\end{algorithmic}
\end{algorithm}

\section{Tensorization of bounded Gaussian mechanism}
In this section we introduce how to tensorize the privacy analysis from the previous section and challenges of its application to DP-SGD. 

\subsection{$L_{\infty}$ bounded support}
We consider the case for truncated Gaussian first. Assume the support set for the bounded Gaussian mechanism in the multi-dimensional case is $\mathcal{B}$. To derive the pdf for truncated Gaussian, we need to calculate the cdf within $\mathcal{B}$, namely $\int_{\mathcal{B}}\phi(\mathbf{x})d\mathbf{x}$. When $\mathcal{B}=\{\mathbf{x}:\|\mathbf{x}\|_2\leq c\}$, the above expression is a $d$-dimensional ball integral over the standard multi-dimensional Gaussian pdf. Neither is it easy to compute its closed form nor are we aware of efficient approximation of it for large $d$. Meanwhile, consider the same $L_2$ bounded support but over the rectified Gaussian with location parameter $\theta$. We will need to compute $\Phi(\theta-\mathbf{x})$ for all $\mathbf{x}$ such that $\|\mathbf{x}\|_2=c$. Deriving its FIL and R\'enyi Divergence is complex and inefficient in practice for large $d$. Hence, instead of using $L_2$ bounded support, we focus on $L_{\infty}$ truncation / rectification in our work for high dimension input. A key benefit of that is we can decompose the $d$-dimensional problem into $d$ scalar problems and account for the privacy budget separately for each coordinate. The privacy cost for the entire input is then the sum of privacy cost calculated at each coordinate.

As a direct result of that, $L_2$ clipping typically used in Gaussian mechanism does not provide us with tight privacy analysis. When $d$ is large, for all coordinates, using the $L_2$ clipping bound as the sensitivity significantly overestimate the privacy budget. Therefore, we use $L_{\infty}$ clipping for multi-dimensional input. 

\subsection{Accounting in multi-dimension}
To calculate the R\'enyi Divergence and FIL for bounded Gaussian mechanism in the multi-dimensional case, we need to integrate the multivariate Gaussian pdf over the bounded support set. With $L_{\infty}$ bounded support and $L_{\infty}$ clipping, we are able to do it at the coordinate level and compose the privacy costs across the coordinates afterwards. Specifically, we proved the following proposition for reducing $d$-dimensional accounting to scalar accounting.
\begin{proposition}
    \label{prop:multidim}
    Let $f(D)=\boldsymbol{\theta}\in\mathbb{R}^d$, $\mathcal{B}=\{\mathbf{x}|\|\mathbf{x}\|_{\infty}\leq a\}$ the bounded support set. Let $S=[-a,a]$. Given $\boldsymbol{\theta}$, the bounded Gaussian mechanism $BG(\cdot)$ satisfies $(\alpha,\epsilon)$-per instance RDP for all where
    \begin{align*}
        \epsilon &= D_{\alpha}\left(\mathcal{N}^B\left(\boldsymbol{\theta},\sigma^2\mathbf{I}_d,\mathcal{B}\right)\|\mathcal{N}^B\left(\boldsymbol{\theta}+c\cdot\mathbf{1}_d,\sigma^2\mathbf{I}_d,\mathcal{B}\right)\right)\\
        &= \sum_jD_{\alpha}\left(\mathcal{N}^B\left(\boldsymbol{\theta}_j,\sigma^2,S\right)\|\mathcal{N}^B\left(\boldsymbol{\theta}_j+c,\sigma^2,S\right)\right)
    \end{align*}
    
    Further, compute $\boldsymbol{\eta}_j:=\eta^B(\boldsymbol{\theta}_j,\sigma^2,S)$ using Lemma \ref{lemma:fil_closed_form} with location parameter $\boldsymbol{\theta}_j$, variance $\sigma^2$, and bounded support $S$, $BG(\cdot)$ satisfies FIL where
    \begin{align*}
        \mathcal{I}_h(D) = J_f^\top \textbf{diag}\left(\left[\boldsymbol{\eta}_1^2,\cdots,\boldsymbol{\eta}_d^2\right]\right)J_f 
    \end{align*}
    where $\textbf{diag}(\cdot)$ denotes the diagonal matrix. 
\end{proposition}

We provide detailed steps for privacy accounting for the bounded Gaussian mechanism in Algorithm~\ref{alg:2}. 

\subsection{Difficulty in subsampling}

Privacy amplification via subsampling is a common technique used in analyzing the privacy cost for mini-batch gradient based optimization algorithm. It has been shown and implemented in multiple works that we can efficiently account the privacy for a subsampled Gaussian mechanism with RDP~\citep{mironov2019r}. Unfortunately, these techniques heavily rely on certain properties of the Gaussian distribution that is not true for bounded Gaussian distribution with $L_{\infty}$ bounded support. In the rest of this section, we will explain the difficulties of analyzing amplification via subsampling for bounded Gaussian mechanism.

\textbf{Setup.} Assume we have two random variables $\mu_{\theta}$ and $\mu_{\theta’}$ sampled from two distribution located at $\theta,\theta’\in\mathbb{R}^d$ respectively. We care about calculating 
\begin{equation}
\label{eq:subsampled_RD}
    D_{\alpha}(\mu_\theta\|(1-q)\mu_\theta+q\mu_{\theta’})\tag{*}
\end{equation}
Suppose $\|\theta-\theta’\|_2=c$. When $d=1$, this quantity could be efficiently approximated by simulating the integral.


\textbf{$\mathcal{N}^b$ is not shift / rotation invariant.} A critical step to calculate (\ref{eq:subsampled_RD}) for the Gaussian mechanism is to apply an affine transformation to $[\theta,\theta']$ so that computing (\ref{eq:subsampled_RD}) is equivalent to computing $D_{\alpha}(\mu_0\|(1-q)\mu_0+qc\mathbf{e}_1)$ where $\mathbf{e}_1$ is the unit vector \citep{mironov2019r}. 
However such equivalence fails under rectified/truncated Gaussian distribution considered in this work because the landscape of both distributions change as you shift or rotate $[\theta,\theta']$.

For the above reasons, we find that it is non-trivial to reduce the calculation of (\ref{eq:subsampled_RD}) for arbitrary large $d$ to a set of 1-dimensional sub problems---requiring integration at every coordinate, which is computationally inefficient for SGD applications where $d$ is large. Therefore, we only consider full batch gradient descent in our DP experiment. We believe efficient accounting for a subsampled bounded Gaussian mechanism is an interesting direction of future work.

\section{Experiment}
We evaluate private SGD with the bounded Gaussian mechanism on multiple datasets to show that enforcing bounded support of the noisy output provides amplified privacy-utility tradeoff for practical learning tasks. 
\subsection{Synthetic data mean estimation}
We first demonstrate how our new privacy analysis for the Gaussian mechanism with bounded support help improve the privacy utility tradeoff for mean estimation problem. 

\textbf{Setup.} We randomly sample $X_1,\cdots,X_n$ from a Gaussian $\mathcal{N}(\mu\mathbf{1}_d,I_d)$ and clip all data to $[-1,1]^d$. To estimate the mean $\mu$ with $X_1,\cdots,X_n$ privately, we take the average of the data and apply output perturbation on top of that: $\hat{\mu}=Perturb(\frac{1}{n}\sum_{i=1}^nX_i)$. Unlike the Gaussian mechanism, the rectified Gaussian mechanism could output a biased estimate when the rectification range is not centered at the true location parameter. In this experiment, we explore how such biasedness affects  privacy amplification via bounded support. We restrict the bounded support set to take the form of $B=[-a,a]^d$ for some $a\in\mathbb{R}$. Hence, at every coordinate, the bounded support centered at 0. Define the \textit{Biasedness} as the distance between the location parameter $\mu$ and the rectification center, in this case, $|\mu|$. We pick $n=900$ and $d=100$ and tune the noise multiplier $\sigma$ and bounded support range at each coordinate $a$. We evaluate the MSE loss and $(2,\epsilon)$-RDP for each hyperparameter combination. Results are averaged over 5 independent trials. 

\begin{figure}[t!]
    \centering
    \begin{subfigure}{0.45\textwidth}
        \centering
        \includegraphics[width=0.99\textwidth,trim=10 10 10 30]{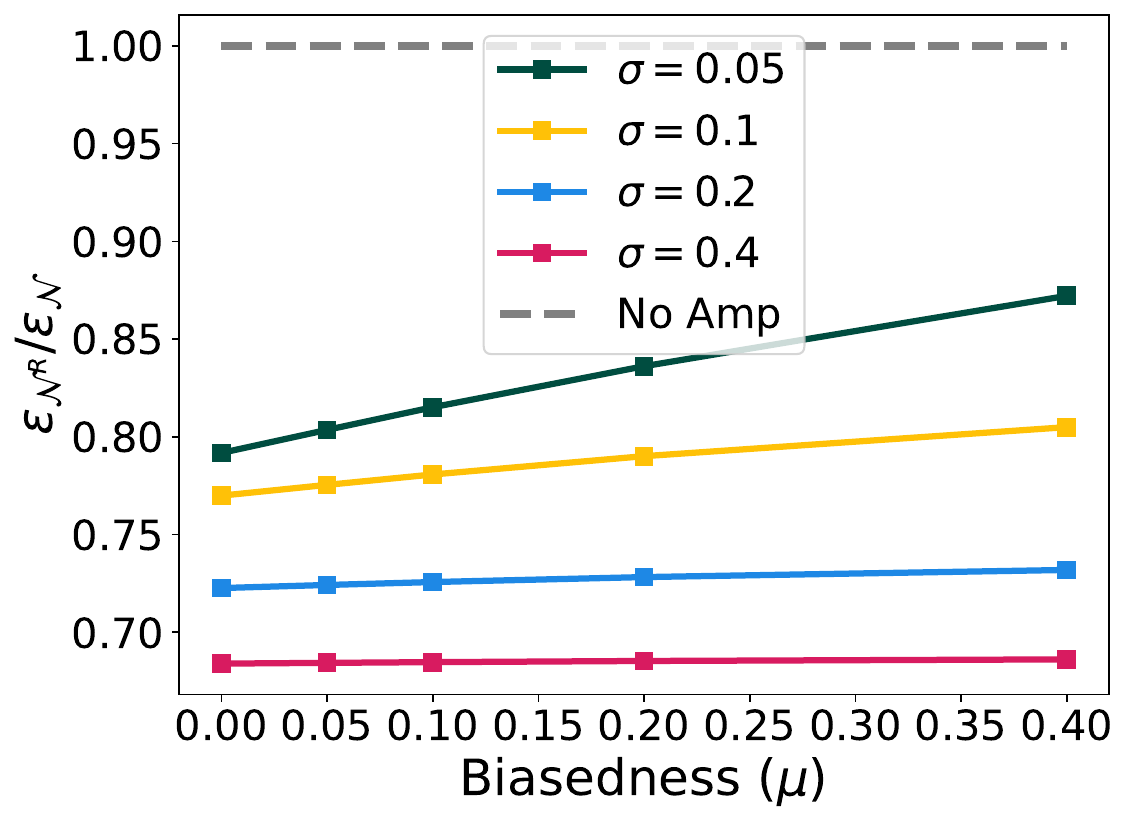}
        \subcaption{\small Biasedness vs Privacy}
    \end{subfigure}\hfill
    \begin{subfigure}{0.45\textwidth}
        \centering
        \includegraphics[width=0.99\textwidth,trim=10 10 10 30]{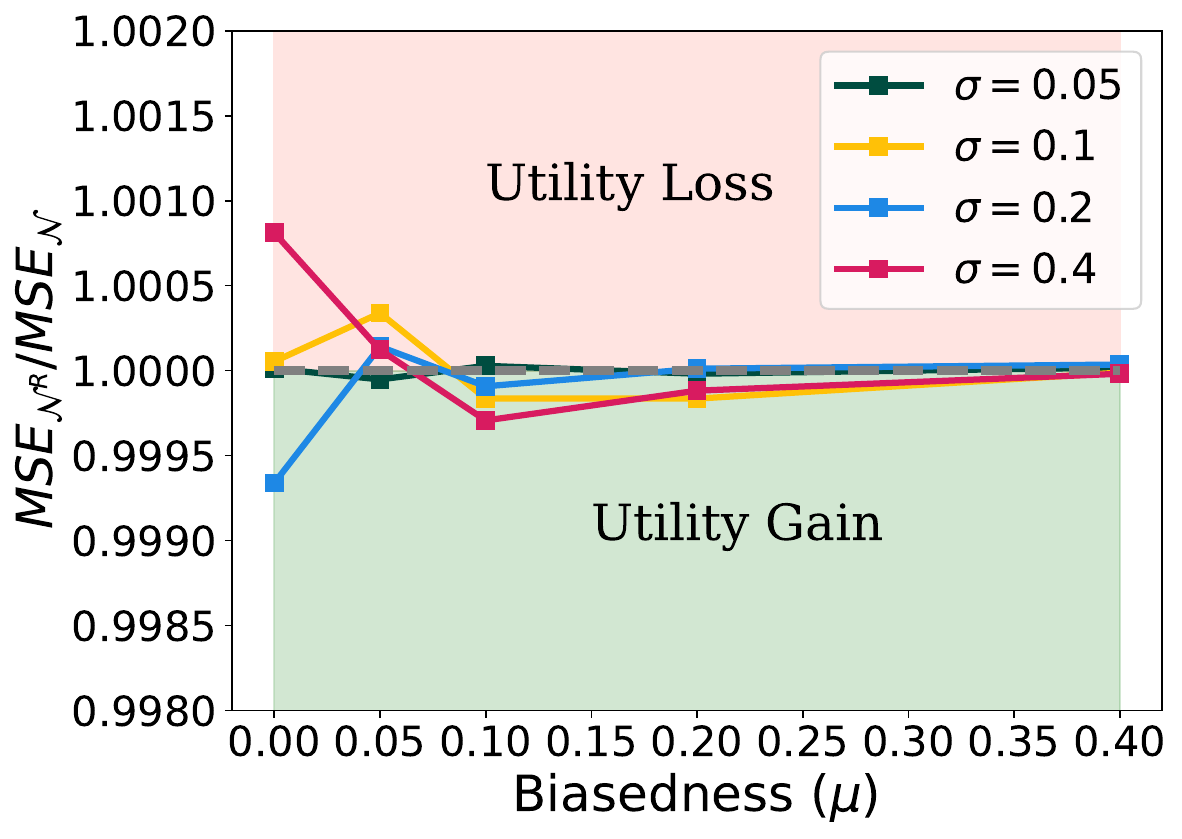}
        \subcaption{\small Biasedness vs Utility}
    \end{subfigure}
    \caption{\small  Biasedness-privacy-utility tradeoffs for synthetic mean estimation. \textit{Left}: Relation between biasedness and privacy amplification. Smaller $\epsilon_{\mathcal{N}^R}/\epsilon_{\mathcal{N}}$ means stronger amplification. \textit{Right}: Relation between biasedness and utility change corresponding to the same hyperparameter combinations shown in the Left figure.} %
    \label{fig:mean_estimation}
\end{figure}

\begin{figure*}[t!]
    \centering
    \begin{subfigure}{0.32\textwidth}
        \centering
        \includegraphics[width=0.99\textwidth,trim=10 10 10 0]{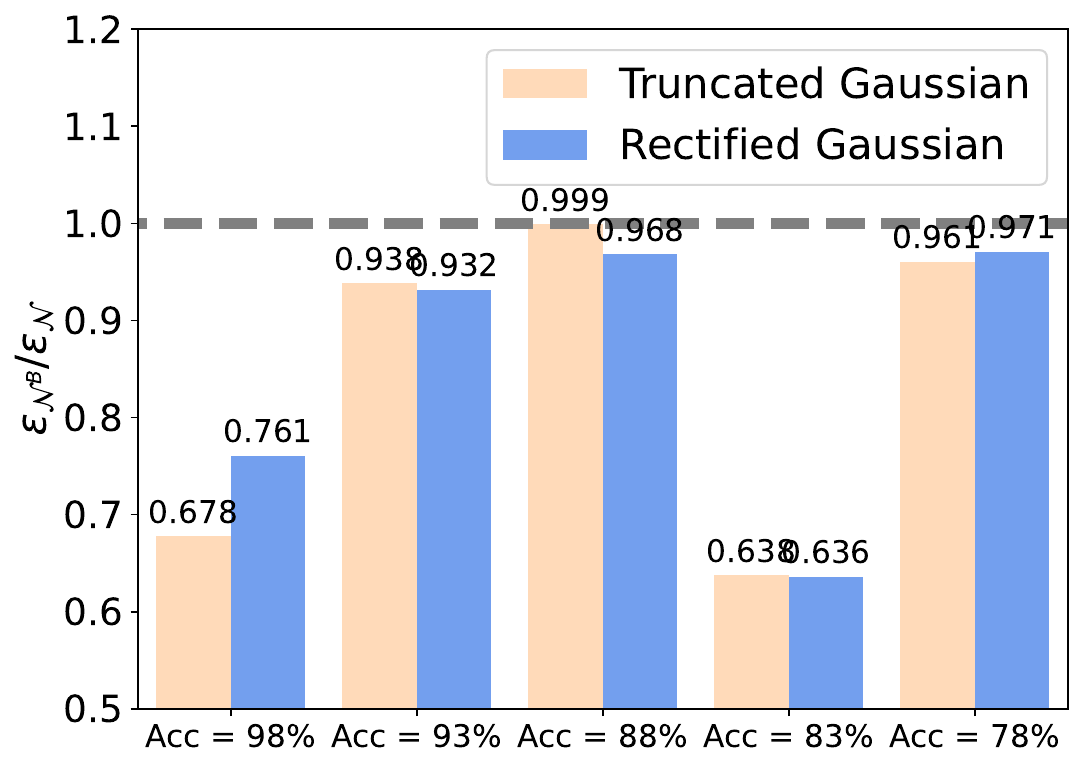}
        \subcaption{CIFAR 10 BeitV2}
        \label{fig:cifar_beit}
    \end{subfigure}\hfill
    \begin{subfigure}{0.32\textwidth}
        \centering
        \includegraphics[width=0.99\textwidth,trim=10 10 10 0]{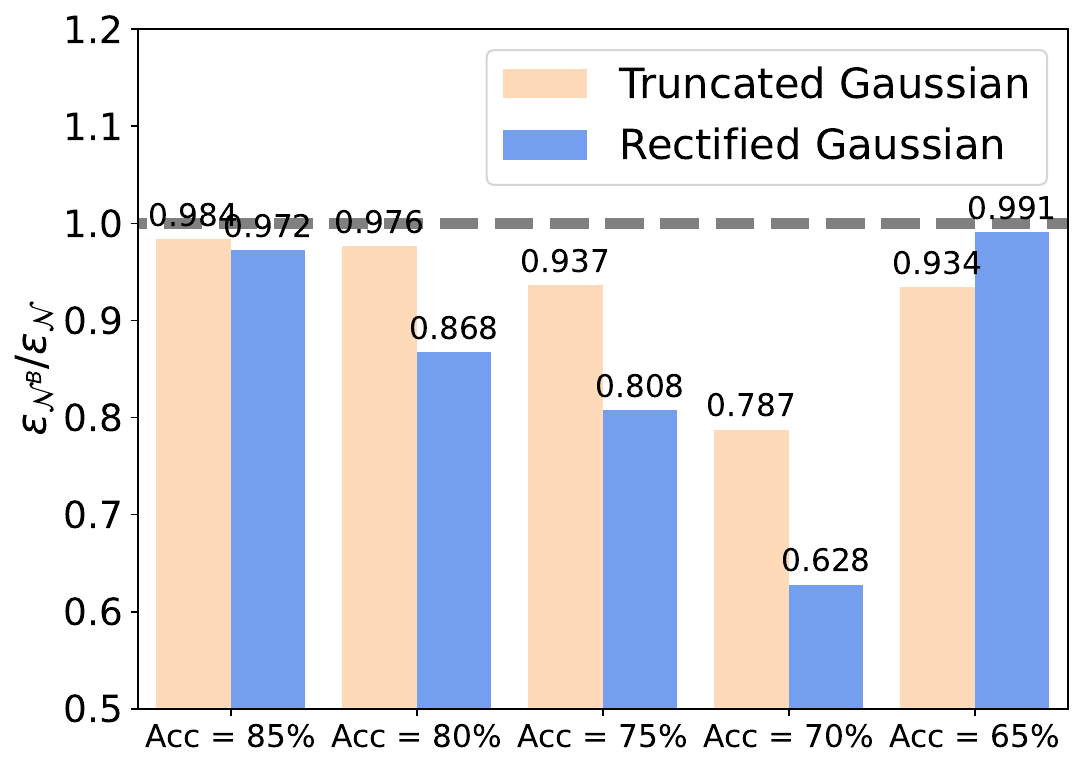}
        \subcaption{CIFAR 100 BeitV2}
        \label{fig:cifar100_beit}
    \end{subfigure}\hfill
    \begin{subfigure}{0.32\textwidth}
        \centering
        \includegraphics[width=0.99\textwidth,trim=10 10 10 0]{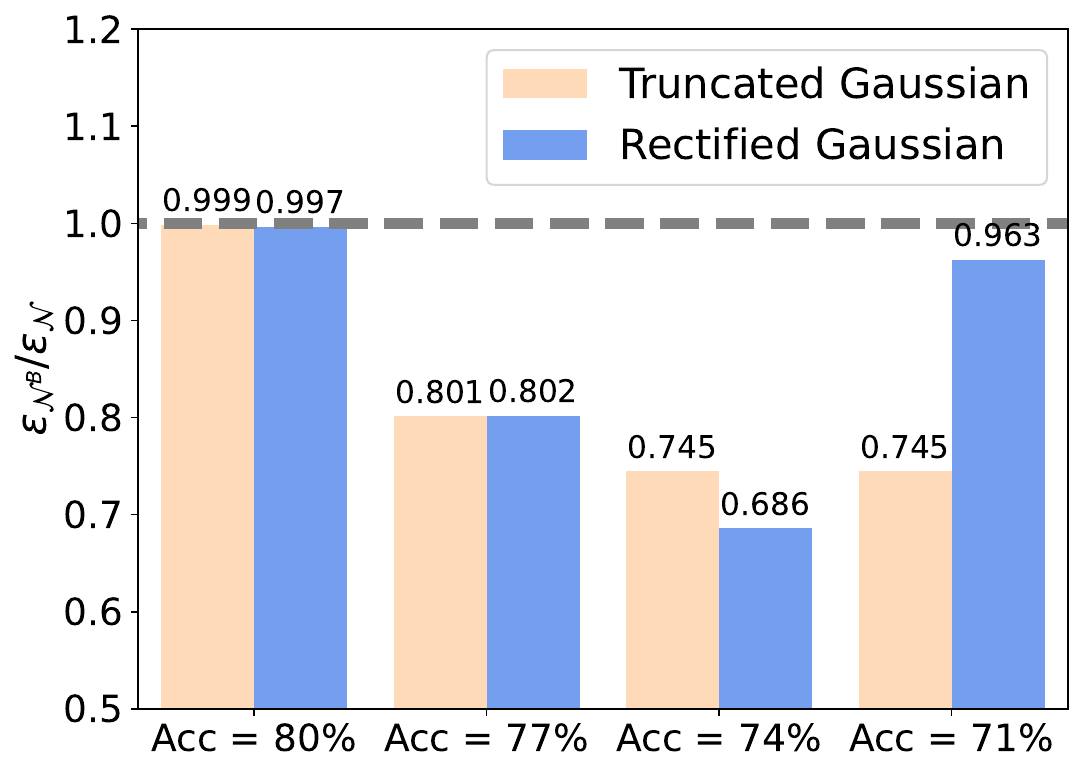}
        \subcaption{OxfordIIIT BeitV2}
        \label{fig:oxford_beit}
    \end{subfigure}\hfill
    \begin{subfigure}{0.32\textwidth}
        \centering
        \includegraphics[width=0.99\textwidth,trim=10 10 10 0]{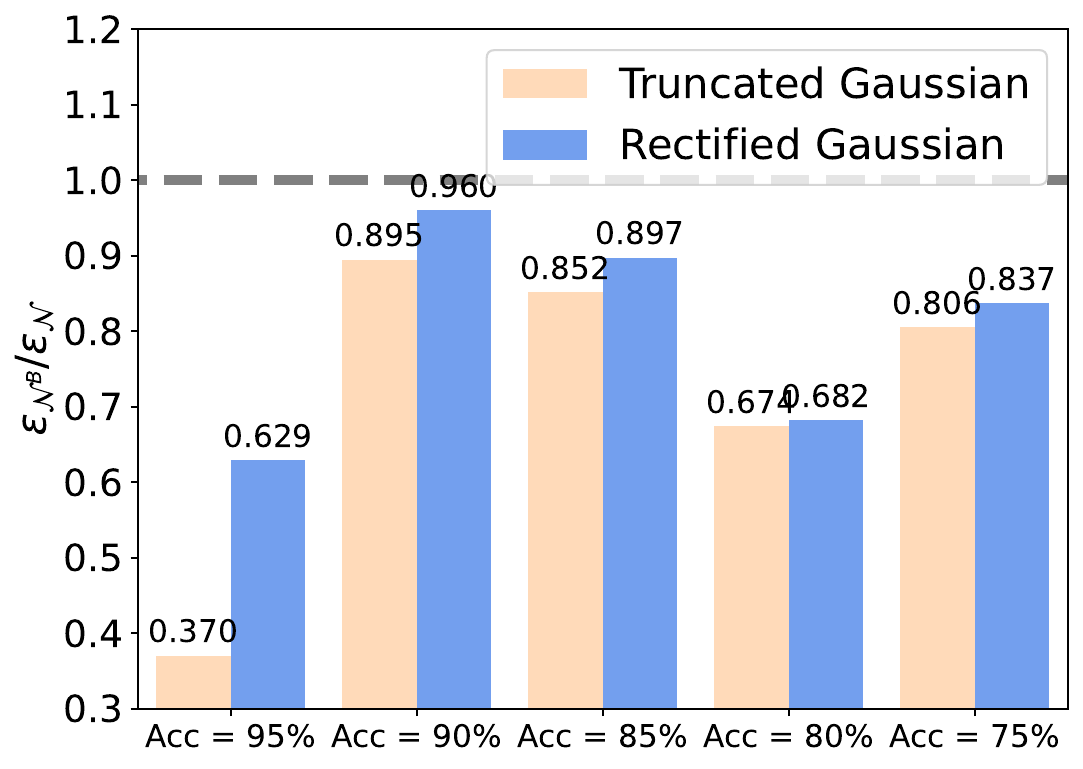}
        \subcaption{CIFAR 10 ViT}
        \label{fig:cifar_vit}
    \end{subfigure}\hfill
    \begin{subfigure}{0.32\textwidth}
        \centering
        \includegraphics[width=0.99\textwidth,trim=10 10 10 0]{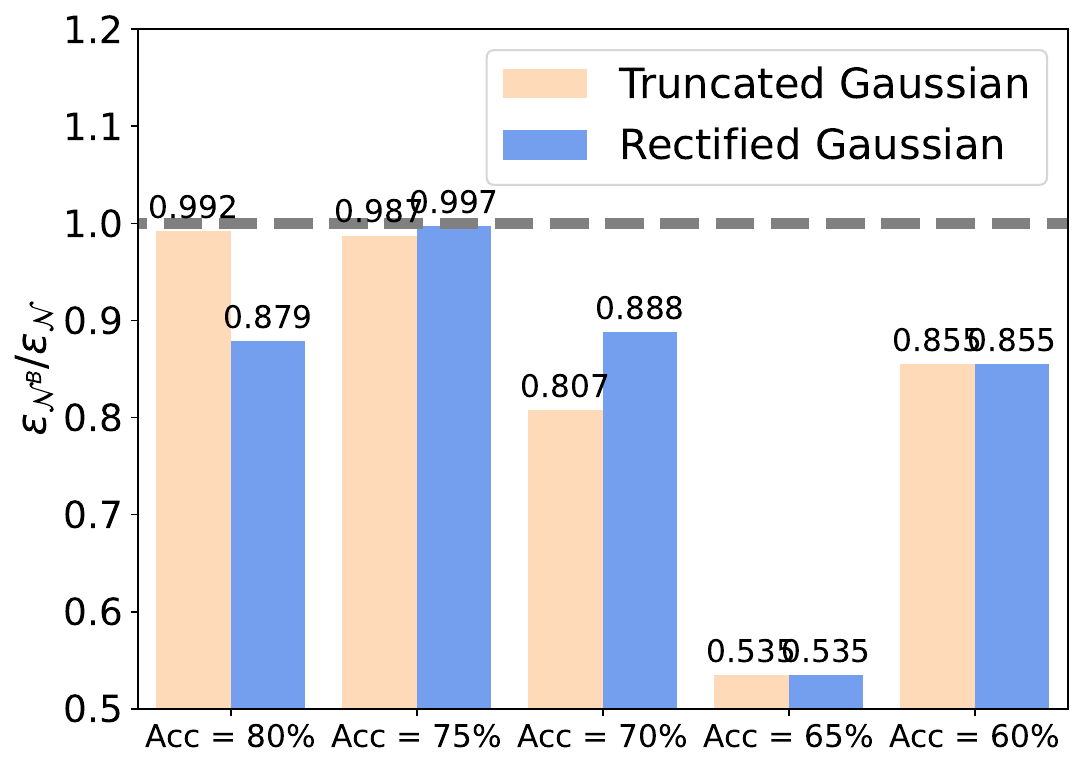}
        \subcaption{CIFAR 100 ViT}
        \label{fig:cifar100_vit}
    \end{subfigure}\hfill
    \begin{subfigure}{0.32\textwidth}
        \centering
        \includegraphics[width=0.99\textwidth,trim=10 10 10 0]{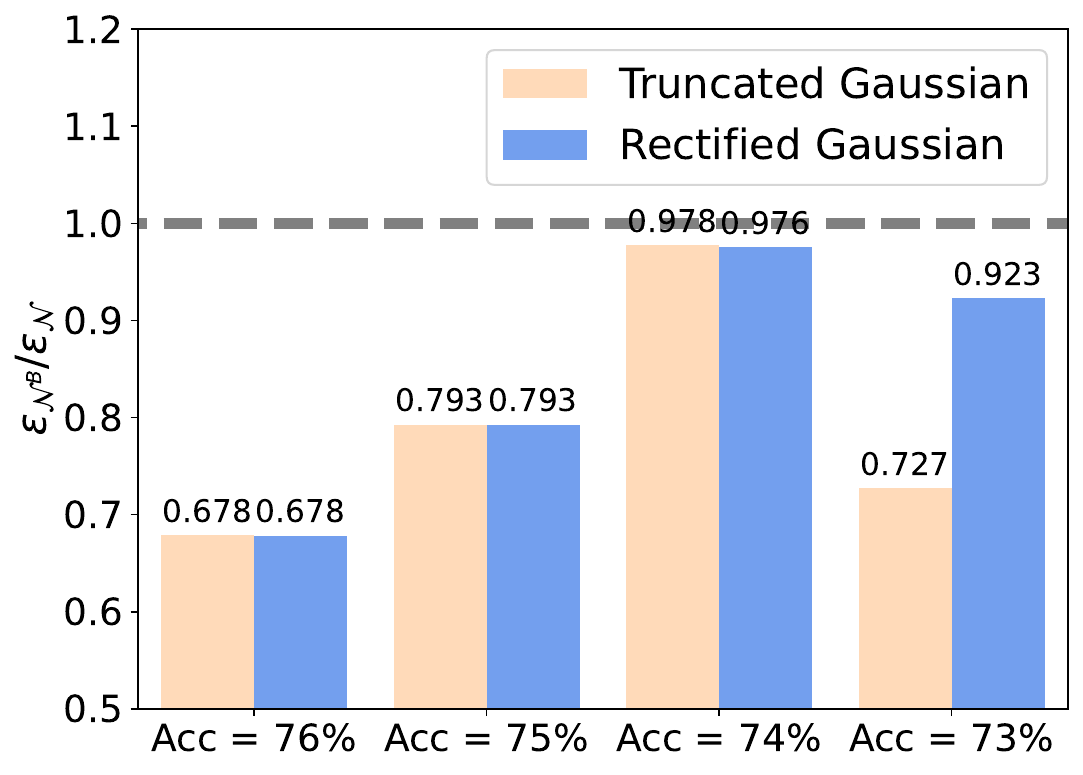}
        \subcaption{OxfordIIIT ViT}
        \label{fig:oxford_vit}
    \end{subfigure}\hfill
    \caption{\small  Comparison between bounded Gaussian mechanism and Gaussian mechanism in terms of pDP-for-all utility tradeoff on three datasets. The $y$-axis is the ratio between bounded Gaussian pDP $\epsilon$ and Gaussian pDP $\epsilon$. The gray dotted line represents no amplification. We limit the true accuracy for both mechanisms to be no more than $1\%$ higher than the target accuracy reported in the figure.} %
    \label{fig:dp_utility_tradeoff}
\end{figure*}
    
    

\textbf{Results}. As demonstrated in Figure \ref{fig:mean_estimation}, in all cases the rectified Gaussian mechanism provides both amplified privacy as well as comparable utility compared to the Gaussian mechanism. When the rectification range is unbiased ($\mu=0$), our method can provide up to 30\% reduction in the privacy cost while only suffering from $<0.5\%$ change in the MSE Loss (e.g. $\sigma=0.4$). Meanwhile, privacy amplification shrinks as the biasedness increases, despite amplified privacy utility tradeoff via bounded support. 

\subsection{pDP utility tradeoff}

\begin{wrapfigure}{r}{0.55\textwidth}
    \centering
    \begin{subfigure}{0.25\textwidth}
        \centering
        \includegraphics[width=0.99\textwidth,trim=10 10 10 30]{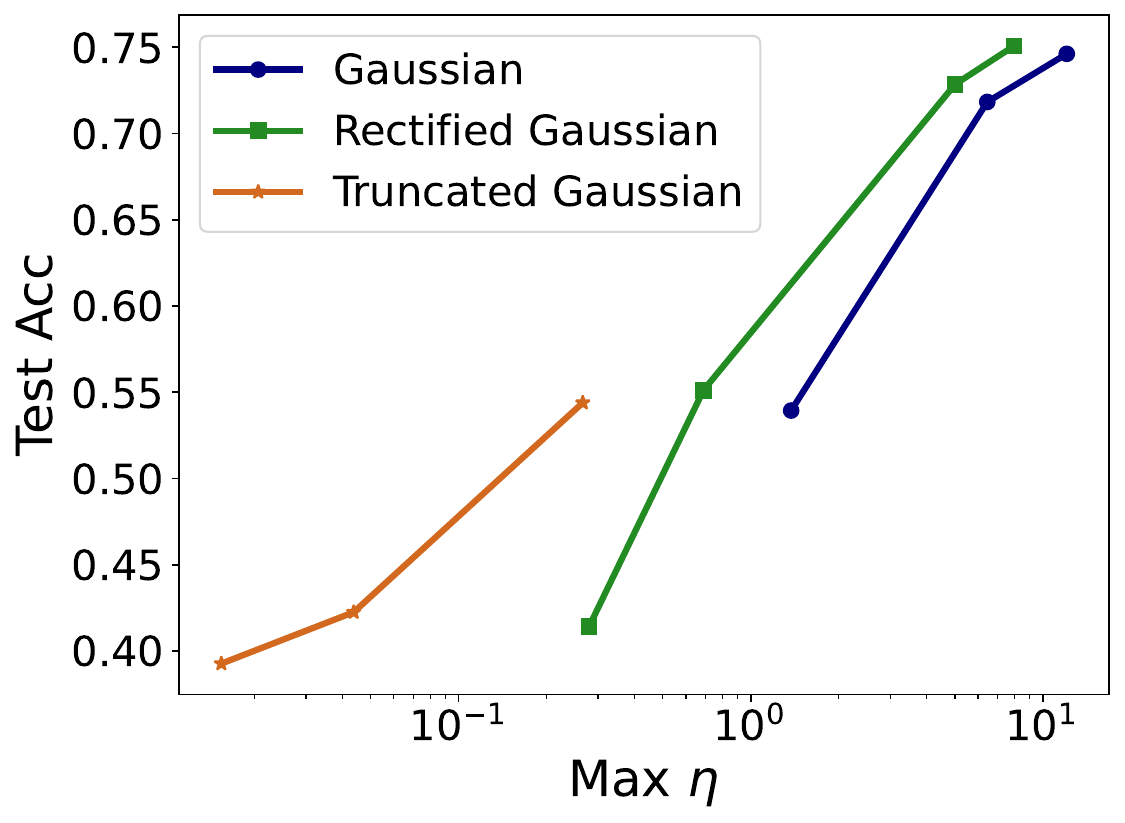}
        \subcaption{\small Full training, $\max\eta$}
        \label{fig:full_train_max_eta}
    \end{subfigure}\hfill
    \begin{subfigure}{0.25\textwidth}
        \centering
        \includegraphics[width=0.99\textwidth,trim=10 10 10 30]{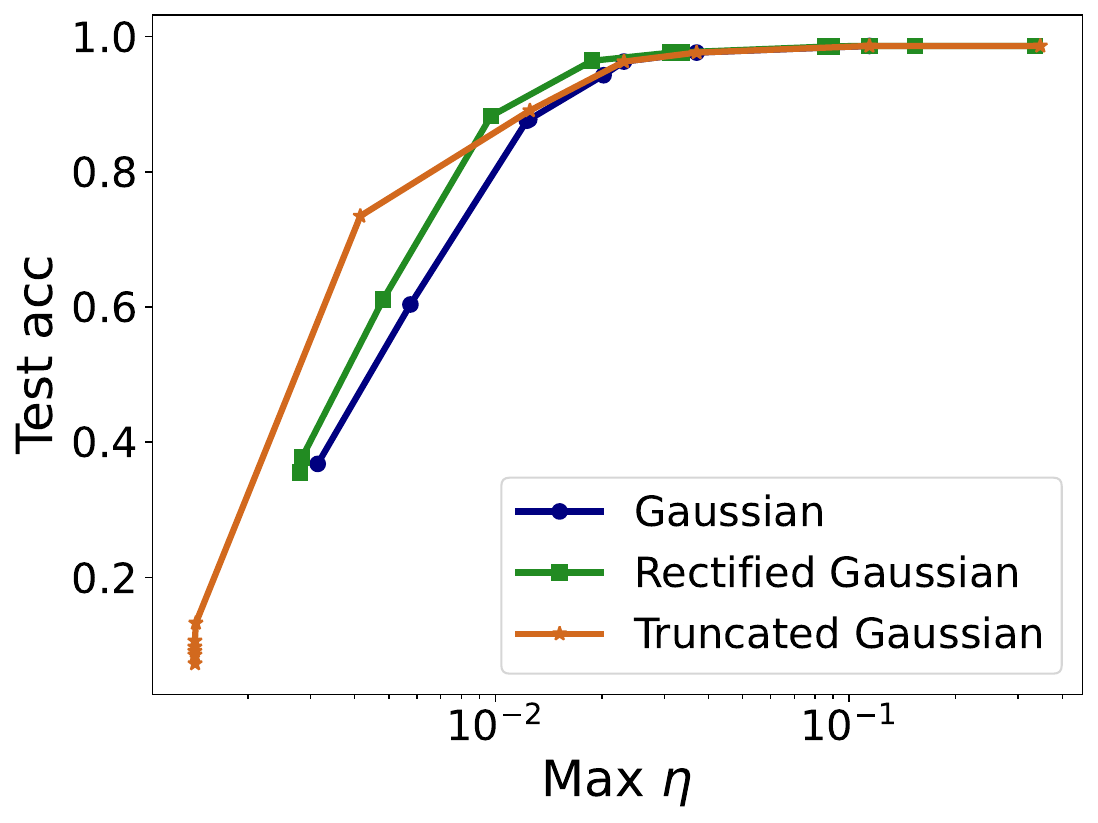}
        \subcaption{\small Linear probing, $\max\eta$}
        \label{fig:linear_probe_max_eta}
    \end{subfigure}
    \begin{subfigure}{0.25\textwidth}
        \centering
        \includegraphics[width=0.99\textwidth,trim=10 10 10 0]{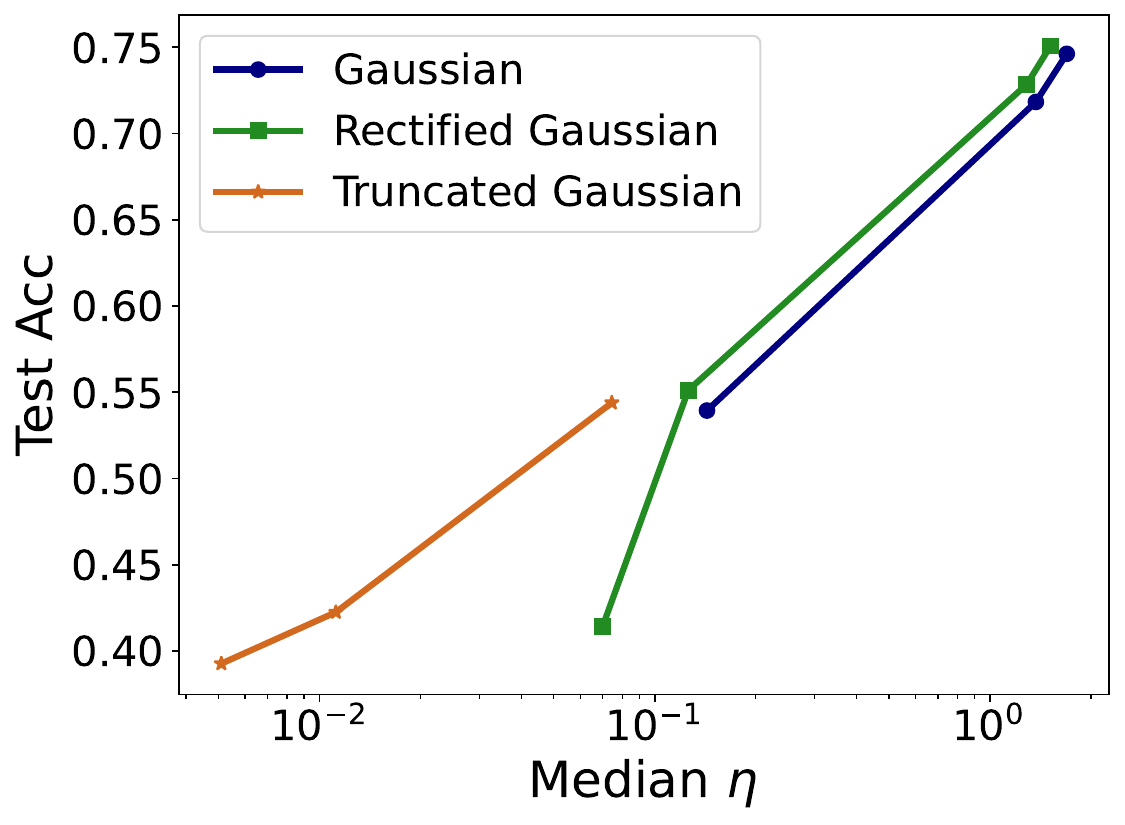}
        \subcaption{\small Full training, median $\eta$}
        \label{fig:full_train_median_eta}
    \end{subfigure}\hfill
    \begin{subfigure}{0.25\textwidth}
        \centering
        \includegraphics[width=0.99\textwidth,trim=10 10 10 0]{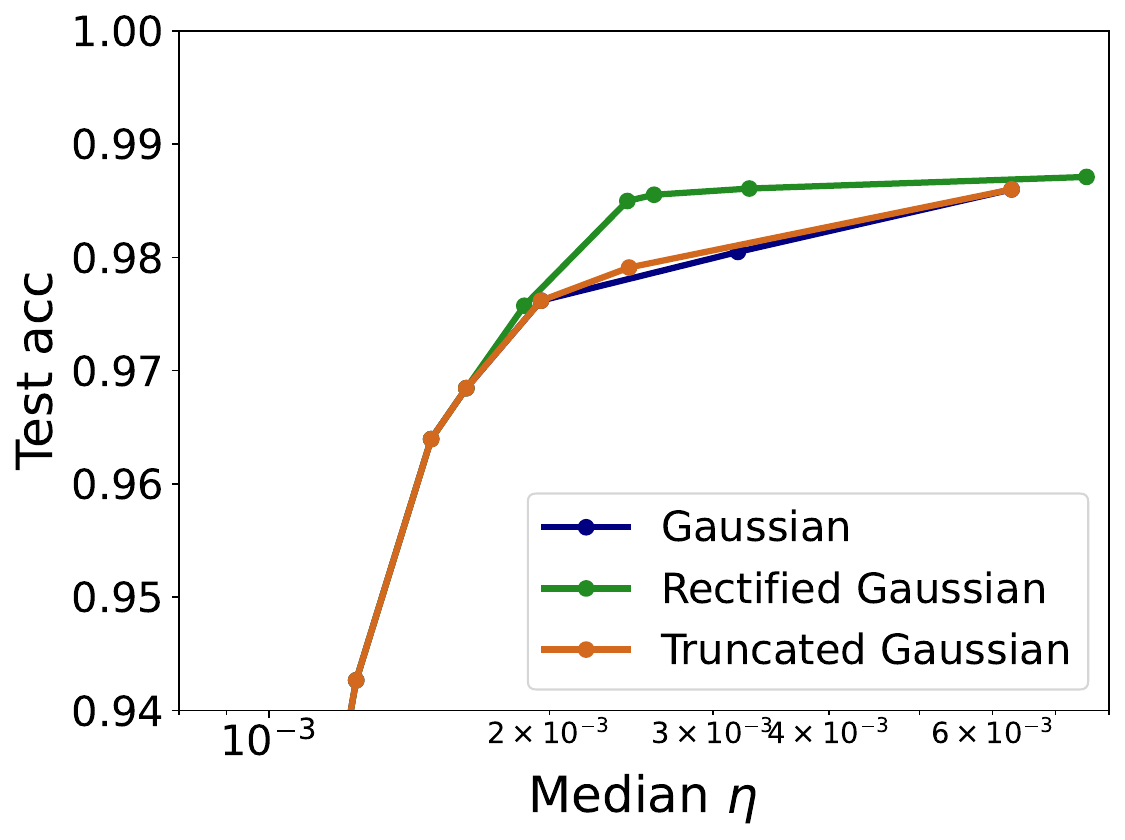}
        \subcaption{\small Linear probing, median $\eta$ }
        \label{fig:linear_probe_median_eta}
    \end{subfigure}
    \caption{\small  Comparison between bounded Gaussian mechanism and Gaussian mechanism in terms of FIL-utility tradeoff.} %
    \label{fig:fil_utility_tradeoff}
\end{wrapfigure}
We evaluate the  pDP-utility tradeoff for the bounded Gaussian mechanism. We consider three image classification tasks: CIFAR10, CIFAR100 \citep{krizhevsky2009learning}, and OxfordIIITPet \citep{parkhi2012oxford}. Following \citet{panda2022dp}, we take a pretrained feature extractor and finetune a linear model on top. Compared to performing private training from scratch, it has been shown that linear probing on a large pretrained model significantly improves the privacy-utility tradeoff \citep{panda2022dp,de2022unlocking}. We perform finetuning on two models pretrained on ImageNet-21k \citep{deng2009imagenet}: beitv2-224 \citep{peng2022beit} and ViT-large-384 \citep{dosovitskiy2020image}. As discussed earlier, we use full batch training for all experiments in this section, which is similarly done by prior works to reduce effective noise and improve utility \citep{panda2022dp}. For fair comparison, we use $L_{\infty}$ clipping for all privacy mechanisms to control the sensitivity. We perform grid search over the all hyperparameters for all mechanisms. Details of hyperparameters can be found in Appendix \ref{appen:hyperparam}.

Results are shown in Figure \ref{fig:dp_utility_tradeoff}. For all datasets and models, we fix the utility requirement and measure the  privacy amplification given that both the bounded Gaussian and Gaussian mechanism have achieved the target test accuracy. In many cases, the bounded Gaussian mechanism significantly reduces the privacy spent without drastically sacrificing the utility. To provide a few examples, finetuning on CIFAR 10 using beitv2 can get 98\% accuracy by reducing $>32\%$ of the privacy cost via truncation (See Figure \ref{fig:dp_utility_tradeoff}\subref{fig:cifar_beit}). Finetuning on CIFAR 100 using ViT-large-384 can get 80\% accuracy by reducing $>12\%$ of the privacy cost via rectification (See Figure \ref{fig:dp_utility_tradeoff}\subref{fig:cifar100_vit}). Note that one can always choose the bounded support set large enough so that the bounded Gaussian mechanism recovers the Gaussian mechanism. Thus, our new method and analysis always enjoy privacy utility trade-offs no worse than the Gaussian baseline.
	     
       

	     
       
\subsection{FIL-utility tradeoff}

Finally, we evaluate the FIL-utility tradeoff of the bounded Gaussian mechanism compared to the vanilla Gaussian mechanism on CIFAR10. Different from the DP accounting, FIL accounting for the bounded Gaussian mechanism supports subsampling. Therefore, we are able to use mini-batch gradient for training the model. We perform two different experiments in this section: \textit{1)} train a WRN \citep{zagoruyko2016wide} from scratch with private SGD; \textit{2)} linear probing of a pretrained beitv2 model with private SGD with full batch. Similar to the DP experiments, we only plot the Pareto frontier for the methods. We plot test accuracy vs. the maximum $\eta$ (Figure \ref{fig:fil_utility_tradeoff}\subref{fig:full_train_max_eta},\subref{fig:linear_probe_max_eta}), whose inverse lower bounds the reconstruction variance for the most vulnerable example. We also plot test accuracy vs. the median $\eta$ (Figure \ref{fig:fil_utility_tradeoff}\subref{fig:full_train_median_eta},\subref{fig:linear_probe_median_eta}) to show the average case performance. The results are shown in Figure \ref{fig:fil_utility_tradeoff}. Both rectified and truncated Gaussian achieve improved FIL utility trade off in both training paradigm. For example, to achieve 75\% accuracy, rectified Gaussian achieves $\max\eta=7.97$, outperforming Gaussian ($\max\eta=12.07$) by $>33\%$.

	     
       

\section{Conclusion \& Future work}
In this work we proposed a novel analysis for two privacy mechanisms that consider bounding the support set for the Gaussian mechanism. We proved that these mechanisms enjoy amplified FIL and per-instance DP guarantees compared to the Gaussian mechanism and can improve privacy-utility trade-offs for private SGD. We also showed that our approaches outperform the Gaussian mechanism on real world private deep learning tasks.
An important consideration in future work may be to enable subsampling for pDP accounting, with a potential option being to subsample the coordinates while performing gradient descent~\citep{chen2023privacy} so that accounting can be done per coordinate.
More generally, we hope future works can build upon this work to further study how compression heuristics can help to improve privacy amplification.

\section*{Broader Impact}
This paper presents work whose goal is to advance the field of privacy preserving machine learning. We aim to understand how applying bounded support on top of the Gaussian mechanism can help improve privacy guarantees for tasks like private gradient descent. However, we note that we do not encourage publishing the per-instance DP privacy loss unconditionally due to its data dependent nature~\citep{wang2019per}. Releasing per-instance DP statistics publicly is an active area of research \citep{redberg2021privately} and our work presents an opportunity for future directions that relate compression and publishable per-instance DP.


\newpage
{\small
\bibliographystyle{abbrvnat}
\bibliography{references}
}

\newpage
\appendix
\onecolumn
\section{Proof of Theorem \ref{th:bgm_fil_amp}}
\label{appen:fil}
We start by showing the the closed form FIL for truncated and rectified Gaussian mechanism (Lemma \ref{lemma:fil_closed_form}).
\begin{proof}[Proof for Lemma \ref{lemma:fil_closed_form}]
We first consider the case when $\mathcal{M}$ is the truncated Gaussian mechanism. Let $p^T(x)$ be the probability density function for the truncated Gaussian, we have
\begin{align*}
    \frac{\partial}{\partial\theta}\log p^T(x) = \frac{x-\theta}{\sigma^2} +\frac{1}{\sigma} \frac{\phi\left(\frac{b-\theta}{\sigma}\right)-\phi\left(\frac{a-\theta}{\sigma}\right)}{\Phi\left(\frac{b-\theta}{\sigma}\right)-\Phi\left(\frac{a-\theta}{\sigma}\right)}.
\end{align*}

\begin{align*}
    \frac{\partial^2}{\partial\theta^2}\log p^T(x) &= -\frac{1}{\sigma^2}+\frac{1}{\sigma}\left(-\frac{1}{\sigma}\frac{\left(\Phi\left(\frac{b-\theta}{\sigma}\right)-\Phi\left(\frac{a-\theta}{\sigma}\right)\right)\left(\phi'\left(\frac{b-\theta}{\sigma}\right)-\phi'\left(\frac{a-\theta}{\sigma}\right)\right)-\left(\phi\left(\frac{b-\theta}{\sigma}\right)-\phi\left(\frac{a-\theta}{\sigma}\right)\right)^2}{\left(\Phi\left(\frac{b-\theta}{\sigma}\right)-\Phi\left(\frac{a-\theta}{\sigma}\right)\right)^2}\right)\\
    &=-\frac{1}{\sigma^2}+\frac{1}{\sigma^2}\frac{\left(\phi\left(\frac{b-\theta}{\sigma}\right)-\phi\left(\frac{a-\theta}{\sigma}\right)\right)^2}{\left(\Phi\left(\frac{b-\theta}{\sigma}\right)-\Phi\left(\frac{a-\theta}{\sigma}\right)\right)^2}-\frac{1}{\sigma^2}\frac{\phi'\left(\frac{b-\theta}{\sigma}\right)-\phi'\left(\frac{a-\theta}{\sigma}\right)}{\Phi\left(\frac{b-\theta}{\sigma}\right)-\Phi\left(\frac{a-\theta}{\sigma}\right)}\\
    &=-\frac{1}{\sigma^2}+\frac{1}{\sigma^2}\frac{\left(\phi\left(\frac{b-\theta}{\sigma}\right)-\phi\left(\frac{a-\theta}{\sigma}\right)\right)^2}{\left(\Phi\left(\frac{b-\theta}{\sigma}\right)-\Phi\left(\frac{a-\theta}{\sigma}\right)\right)^2}-\frac{1}{\sigma^2}\frac{\frac{a-\theta}{\sigma}\phi\left(\frac{a-\theta}{\sigma}\right)-\frac{b-\theta}{\sigma}\phi\left(\frac{b-\theta}{\sigma}\right)}{\Phi\left(\frac{b-\theta}{\sigma}\right)-\Phi\left(\frac{a-\theta}{\sigma}\right)}.
\end{align*}
Therefore, \begin{align*}
    \mathcal{I}_{\mathcal{M}}(\theta) &=\mathbb{E}_{x}\left[-\frac{\partial^2}{\partial\theta^2}\log p^T(x)\right]\\
    &= \frac{1}{\sigma^2}-\frac{1}{\sigma^2}\frac{\left(\phi\left(\frac{b-\theta}{\sigma}\right)-\phi\left(\frac{a-\theta}{\sigma}\right)\right)^2}{\left(\Phi\left(\frac{b-\theta}{\sigma}\right)-\Phi\left(\frac{a-\theta}{\sigma}\right)\right)^2}+\frac{1}{\sigma^2}\frac{\frac{a-\theta}{\sigma}\phi\left(\frac{a-\theta}{\sigma}\right)-\frac{b-\theta}{\sigma}\phi\left(\frac{b-\theta}{\sigma}\right)}{\Phi\left(\frac{b-\theta}{\sigma}\right)-\Phi\left(\frac{a-\theta}{\sigma}\right)}.
\end{align*}
Now we follow the proof from \citet{hannun2021measuring} to compute $\mathcal{I}_{\mathcal{M}}(D)$. Let $\mathbf{H}$ be the second-order derivative of $f(D)$: $\mathbf{H}_{ijk}=\frac{\partial^2f_k}{\partial D_i\partial D_j}$
\begin{align*}
    \mathcal{I}_{\mathcal{M}}(D) &= \mathbb{E}\left[-\nabla_D^2\log p^T(x|D)\right]\\
    &= \mathbb{E}\left[-J_f^\top\frac{\partial^2}{\partial\theta^2}\log p^T(x)J_f-\mathbf{H}\frac{\partial}{\partial\theta}\log p^T(x) \right]\\
    &= J_f^\top \mathcal{I}_{\mathcal{M}}(\theta) J_f-\mathbf{H}\mathbb{E}\left[\frac{\partial}{\partial\theta}\log p^T(x) \right]\\
    &= J_f^\top \mathcal{I}_{\mathcal{M}}(\theta) J_f-\mathbf{H}\mathbb{E}\left[\frac{x-\theta}{\sigma^2} +\frac{1}{\sigma} \frac{\phi\left(\frac{b-\theta}{\sigma}\right)-\phi\left(\frac{a-\theta}{\sigma}\right)}{\Phi\left(\frac{b-\theta}{\sigma}\right)-\Phi\left(\frac{a-\theta}{\sigma}\right)}\right]\\
        &= J_f^\top \mathcal{I}_{\mathcal{M}}(\theta) J_f-\mathbf{H}\left(\frac{\theta+\frac{\phi\left(\frac{a-\theta}{\sigma}\right)-\phi\left(\frac{b-\theta}{\sigma}\right)}{\Phi\left(\frac{b-\theta}{\sigma}\right)-\Phi\left(\frac{a-\theta}{\sigma}\right)}\sigma-\theta}{\sigma^2}+\frac{1}{\sigma} \frac{\phi\left(\frac{b-\theta}{\sigma}\right)-\phi\left(\frac{a-\theta}{\sigma}\right)}{\Phi\left(\frac{b-\theta}{\sigma}\right)-\Phi\left(\frac{a-\theta}{\sigma}\right)}\right)\\
        &=J_f^\top \mathcal{I}_{\mathcal{M}}(\theta) J_f.
\end{align*}
This concludes our derivation of FIL for the truncated Gaussian mechanism.

When $\mathcal{M}$ is the rectified Gaussian mechanism, let $p^R(x)$ be the probability density function for the rectified Gaussian, we have
\begin{align*}
    \frac{\partial}{\partial\theta}\log p^R(x) &= U(x;a,b)\frac{x-\theta}{\sigma^2}-\frac{\mathbf{1}_{x=a}}{\sigma}\frac{\phi\left(\frac{a-\theta}{\sigma}\right)}{\Phi\left(\frac{a-\theta}{\sigma}\right)}+\frac{\mathbf{1}_{x=b}}{\sigma}\frac{\phi\left(\frac{\theta-b}{\sigma}\right)}{\Phi\left(\frac{\theta-b}{\sigma}\right)}
\end{align*}
\begin{align*}
    \frac{\partial^2}{\partial\theta^2}\log p^R(x) &= -U(x;a,b)\frac{1}{\sigma^2}+\frac{\mathbf{1}_{x=a}}{\sigma^2}\frac{\Phi\left(\frac{a-\theta}{\sigma}\right)\phi'\left(\frac{a-\theta}{\sigma}\right)-\phi^2\left(\frac{a-\theta}{\sigma}\right)}{\Phi\left(\frac{a-\theta}{\sigma}\right)^2}+\frac{\mathbf{1}_{x=b}}{\sigma^2}\frac{\Phi\left(\frac{\theta-b}{\sigma}\right)\phi'\left(\frac{\theta-b}{\sigma}\right)-\phi^2\left(\frac{\theta-b}{\sigma}\right)}{\Phi\left(\frac{\theta-b}{\sigma}\right)^2}
\end{align*}
The fisher information is thus:
\begin{align*}
    \mathcal{I}_{\mathcal{M}}(\theta) &= \mathbb{E}_{x}\left[-\frac{\partial^2}{\partial\theta^2}\log p^R(x)\right]\\
    &= \mathbb{E}_{x}\left[U(x;a,b)\frac{1}{\sigma^2}-\frac{\mathbf{1}_{x=a}}{\sigma^2}\frac{\Phi\left(\frac{a-\theta}{\sigma}\right)\phi'\left(\frac{a-\theta}{\sigma}\right)-\phi^2\left(\frac{a-\theta}{\sigma}\right)}{\Phi\left(\frac{a-\theta}{\sigma}\right)^2}-\frac{\mathbf{1}_{x=b}}{\sigma^2}\frac{\Phi\left(\frac{\theta-b}{\sigma}\right)\phi'\left(\frac{\theta-b}{\sigma}\right)-\phi^2\left(\frac{\theta-b}{\sigma}\right)}{\Phi\left(\frac{\theta-b}{\sigma}\right)^2}\right]\\
    &= \frac{\Phi\left(\frac{\theta-b}{\sigma}\right) - \Phi\left(\frac{\theta-a}{\sigma}\right)}{\sigma^2}-\frac{\phi'\left(\frac{a-\theta}{\sigma}\right)}{\sigma^2}+\frac{\phi^2\left(\frac{a-\theta}{\sigma}\right)}{\sigma^2\Phi\left(\frac{a-\theta}{\sigma}\right)}-\frac{\phi'\left(\frac{\theta-b}{\sigma}\right)}{\sigma^2}+\frac{\phi^2\left(\frac{\theta-b}{\sigma}\right)}{\sigma^2\Phi\left(\frac{\theta-b}{\sigma}\right)}\\
    &= \frac{\Phi\left(\frac{\theta-b}{\sigma}\right) - \Phi\left(\frac{\theta-a}{\sigma}\right)}{\sigma^2}+\frac{\frac{a-\theta}{\sigma}\phi\left(\frac{a-\theta}{\sigma}\right)}{\sigma^2}+\frac{\phi^2\left(\frac{a-\theta}{\sigma}\right)}{\sigma^2\Phi\left(\frac{a-\theta}{\sigma}\right)}+\frac{\frac{\theta-b}{\sigma}\phi\left(\frac{\theta-b}{\sigma}\right)}{\sigma^2}+\frac{\phi^2\left(\frac{\theta-b}{\sigma}\right)}{\sigma^2\Phi\left(\frac{\theta-b}{\sigma}\right)}\\
    &= \frac{1}{\sigma^2}\left(\frac{\phi^2\left(\frac{a-\theta}{\sigma}\right)}{\Phi\left(\frac{a-\theta}{\sigma}\right)}+\frac{\phi^2\left(\frac{\theta-b}{\sigma}\right)}{\Phi\left(\frac{\theta-b}{\sigma}\right)}+\Phi\left(\frac{\theta-b}{\sigma}\right) - \Phi\left(\frac{\theta-a}{\sigma}\right)+\frac{a-\theta}{\sigma}\phi\left(\frac{a-\theta}{\sigma}\right) - \frac{b-\theta}{\sigma}\phi\left(\frac{b-\theta}{\sigma}\right)\right).
\end{align*}
Similar to the proof for truncated gaussian, we have
\begin{align*}
    \mathcal{I}_{\mathcal{M}}(D) &= \mathbb{E}\left[-\nabla_D^2\log p^R(x|D)\right]\\
    &= J_f^\top \mathcal{I}_{\mathcal{M}}(\theta) J_f-\mathbf{H}\mathbb{E}\left[\frac{\partial}{\partial\theta}\log p^R(x) \right]\\
    &= J_f^\top \mathcal{I}_{\mathcal{M}}(\theta) J_f-\mathbf{H}\mathbb{E}\left[U(x;a,b)\frac{x-\theta}{\sigma^2}-\frac{\mathbf{1}_{x=a}}{\sigma}\frac{\phi\left(\frac{a-\theta}{\sigma}\right)}{\Phi\left(\frac{a-\theta}{\sigma}\right)}+\frac{\mathbf{1}_{x=b}}{\sigma}\frac{\phi\left(\frac{\theta-b}{\sigma}\right)}{\Phi\left(\frac{\theta-b}{\sigma}\right)}\right]\\
        &=J_f^\top \mathcal{I}_{\mathcal{M}}(\theta) J_f.
\end{align*}
\end{proof}

First note that for gaussian mechanism, $\eta=\frac{1}{\sigma^2}$ \citep{hannun2021measuring}. Now we are ready to show the FIL amplification via bounded support. $\eta^R\leq\eta$ follows directly from the following post-processing property of FIL.

\begin{lemma}[Theorem 2.86 from \citet{schervish2012theory}]
    Let $T=t(X)$ be a statistics. Then the FIM of T and X satisfies $\mathcal{I}_T(\theta)\leq\mathcal{I}_X(\theta)$.
\end{lemma}

Now we focus on showing $\eta^T\leq\eta$. W.L.O.G, we assume $\sigma=1$ for the rest of the proof. We need the following lemma to prove Theorem \ref{th:bgm_fil_amp}.
\begin{lemma}
\label{lemma:concave}
    Let $f(x)=\log(\Phi(b-x)-\Phi(a-x))$, we have $f(x)$ is concave.
\end{lemma}

\begin{proof}[Proof for Lemma \ref{lemma:concave}]
We first prove the following lemma.
\newcommand{\sign}{\mathrm{sign}}
\newcommand{\erf}{\mathrm{erf}}
\begin{lemma}\label{lem:derivatives}
Let $f\colon R\to R$ be a continuous and thrice differentiable function with continuous derivatives and $f(0)=0$, $f'(0)=0$, $f''(0)>0$, $\lim_{x\to \infty}f(x)\geq 0$, $\lim_{x\to \infty}\sign(f’(x))=-1$.  Also, assume $f’(x) = r(x)t(x)$ such that $r(x)>0$ and $t(x)$ is strictly concave for $x>0$. Then we have $f(x)\geq 0$ for all $x>0$. 
\end{lemma}
\begin{proof}[Proof for Lemma \ref{lem:derivatives}]
Based on the conditions on the limit in the infinity, there should exist a point $x_\infty$ such that the for all $x\geq x_\infty$ $f(x_\infty)\geq 0$ and $f'(x_\infty)\leq 0$.
We prove the statement by contradiction. Assume there exist a point $0<x_{-}<x_\infty$ such that $f(x_{-})<0$.  Since $f'(0)=0$ and $f''(0)>0$, we have $\lim_{x\to 0^+} \sign(f'(x)) =1$ and $\lim_{x\to 0^+} \sign(f(x)) =1$. Now, since the function is positive in the right neighborhood of $0$, and it is negative at $x_{-}$, the gradient should become negative in at least one point $0<x_1<x_{-}$. On the other hand, since the function is negative in $x_-$ and positive in $x_{\infty}$, the gradient should get positive in at least one point $x_-<x_2< x_\infty$. Now, considering the sign of $f'$, we have $f'(0)>0$, $f'(x_1)<0$, $f'(x_2)>0$ and $f'(x_\infty)<0$. Since $f'$ is changing sign 3 times, it should at least have $3$ roots as well. This means that $t(x)$ should have three roots. However, we know that $t$ is strictly concave and cannot have three roots. Now we prove $g_c(x)$ is concave for all $c>0$.
\end{proof}
Now we use this lemma to prove the statement. Let us take the second gradient of function $f(x)=\log(\Phi(b-x) - \Phi(a-x)).$ For simplicity, let us define an alternative function $g_c(x) = \log(\erf(x+c) - \erf(x))$. We can write $f(x) =g_{\frac{b-a}{\sqrt{2}}}(\frac{a-x}{\sqrt{2}}) -\log(2)$. Since $g$ is a linear transformation of $f$, proving the concavity of $g$ is equivalent to proving the concavity of $f$. From now on, we use $g$ instead of $g_c$. For the second derivative of $g$ we have:
\begin{align*}\frac{d^2g(x)}{dx^2} &=\frac{-4 (e^{-x^2} - e^{-(c + x)^2})^2}{\pi (\erf(x) - \erf(c + x))^2}\\
&+ \frac{-4\sqrt{\pi} (e^{-x^2} x - e^{-(c + x)^2} (c + x)) (\erf(x) - \erf(c + x))}{\pi(\erf(x) - \erf(c + x))^2}
\end{align*}
We need to prove this quantity is negative for all $x$ and $c$. We only focus on the numerator and prove that the following quantity is positive:
\begin{align*}
h(x) &= (e^{-x^2} - e^{-(c + x)^2})^2  \\
&+ \sqrt{\pi} (e^{-x^2} x - e^{-(c + x)^2} (c + x)) (\erf(x) - \erf(c + x))
\end{align*}
We can only focus on the case that $x$ is positive as the function is symmetric around $x=-c/2$ and for the case were $-c/2<x<0$, we know that $h(x)$ is trivially positive because both terms are positive.

Once again we change the function for simplicity. Let $$r_y(x) = (e^{-x^2} - e^{-y^2})^2
+ \sqrt{\pi} (e^{-x^2} x - e^{-y^2}y) (\erf(x) - \erf(y)).$$
Hence, we have:
$$h(x) = r_x(c+x).$$
We now prove for any $x,c>0$, $ r_x(c+x)\geq 0$. We invoke Lemma \ref{lem:derivatives} to prove this. consider $f(c)=r_x(c+x)$ as a function of $c$. We have $f(0)=0$. We can also observe that $\lim_{c\to \infty} f(c) = e^{-2x^2} -\sqrt{\pi}xe^{-x^2}(1-\erf(x)).$ And this quantity is positive for all $x$, simply because the gradient of $e^{-x^2}-\sqrt{\pi} x(1-\erf(x))$ is equal to $\sqrt{\pi}(\erf(x)-1)$ which is always negative. Therefore, the supremum happens in the limit when pushing $x$ to infinity, which is $0$. Now let us look at the gradient of $f$, which we denote by $f'$. We have
\begin{align*}
f'(c)&= \sqrt{\pi} (2 e^{-(x + c)^2} (x + c)^2 - e^{-(x + c)^2}) (\erf(x) - \erf(x + c))\\
&+ 4 e^{-(x + c)^2} (e^{-x^2} - e^{-(x + c)^2}) (x + c)\\
&- 2 e^{-(x + c)^2} (e^{-x^2} x - e^{-(x + c)^2} (x + c))\\
&= e^{-(x+c)^2}\Big((2(x+c)^2 -1 )(\erf(x)-\erf(x+c))\\
&~~~~~~~~~~~~+4(e^{-x^2} - e^{-(x + c)^2}) (x + c)\\
&~~~~~~~~~~~~- 2 (e^{-x^2} x - e^{-(x + c)^2} (x + c))\Big)
\end{align*}
It is easy to see that $\lim_{c\to \infty} f'(c) = 0$ because of the $e^{-(x+c)^2}$ factor. Also we can observe that the $\lim_{c\to \infty}\sign(f'(c))=1$. Now note that we have already separated $f'$ in the form of $f'(c)=t(c)q(c)$ where $t(c)=e^{-(x+c)^2}>0$. So, the only step left is to show that $q(c)$ is concave. To show this, we take the second gradient of $q(c)=\sqrt{\pi}(2(x+c)^2 -1 )(\erf(x)-\erf(x+c))+
4(e^{-x^2} - e^{-(x + c)^2}) (x + c)
- 2 (e^{-x^2} x - e^{-(x + c)^2} (x + c)$. We have
$$q''(c) = 4 \sqrt{\pi} (\erf(x) - \erf(c + x)) - 8 (c+x) e^{-(c + x)^2}.$$
Note that this is negative because $(\erf(x) - \erf(c+x)<0$  (Since $c+x>x$). Also $-8(c+x)e^{-(c+x)^2}<0$ since $c+x>0$. 
\end{proof} 

Now we have all the ingredients for proving Theorem \ref{th:bgm_fil_amp}.

\begin{proof}[Proof for Theorem \ref{th:bgm_fil_amp}]
    By Lemma \ref{lemma:concave}, $f''(x)<0$ everywhere. Note that when $\mathcal{M}$ is truncated gaussian mechanism, $\mathcal{I}_{\mathcal{M}}(\theta)=1+f''(x)<1$. Note that $\eta^2=\frac{1}{\sigma^2}=1$. Hence, we completed the proof that $\eta^T\leq \eta$.
\end{proof}

\section{Derivation of R\'enyi Divergence of bounded Gaussian distribution}
\label{appen:dp_derivation}
\begin{proof}

Consider the case of truncated Gaussian:
\begin{align*}
    &D_{\alpha}(\mathcal{N}^T(\theta,\sigma^2,\mathcal{B})\|\mathcal{N}^T(\theta+c,\sigma^2,\mathcal{B}))\\ 
    &=\frac{1}{\alpha-1}\log\int_{-a}^{a}\frac{1}{\sigma}\left(\frac{\phi\left(\frac{x-(\theta+c)}{\sigma}\right)}{\Delta(\theta+c)}\right)^{1-\alpha}\left(\frac{\phi\left(\frac{x-\theta}{\sigma}\right)}{\Delta(\theta)}\right)^{\alpha}dx\\
    &=\frac{1}{\alpha-1}\log\int_{-a}^{a}\frac{1}{\sigma}\phi\left(\frac{x-\theta}{\sigma}\right)^{\alpha}\phi\left(\frac{x-\theta-c}{\sigma}\right)^{1-\alpha}\left(\frac{\Delta(\theta+c)}{\Delta(\theta)}\right)^{\alpha-1}\left(\frac{1}{\Delta(\theta)}\right)dx\\
    &=\frac{1}{\alpha-1}\log\Bigg(\left(\frac{\Delta(\theta+c)}{\Delta(\theta)}\right)^{\alpha-1}\left(\frac{1}{\Delta(\theta)}\right)\int_{-a}^{a}\frac{1}{\sigma}\phi\left(\frac{x-\theta}{\sigma}\right)^{\alpha}\phi\left(\frac{x-\theta-c}{\sigma}\right)^{1-\alpha}dx\Bigg)
\end{align*}

Now we focus on the integration part:
\begin{align*}
    &\int_{-a}^a\frac{1}{\sigma}\phi\left(\frac{x-\theta}{\sigma}\right)^{\alpha}\phi\left(\frac{x-\theta-c}{\sigma}\right)^{1-\alpha}dx\\
    &=\frac{1}{\sqrt{2\pi}\sigma}\int_{-a}^a\exp\left(-\frac{1}{2\sigma^2}\left(\alpha(x-\theta)^2+(1-\alpha)(x-\theta-c)^2\right)\right)dx\\
    &=\frac{1}{\sqrt{2\pi}\sigma}\int_{-a}^a\exp\left(-\frac{1}{2\sigma^2}\left(x^2-2(\theta+(1-\alpha)c)x+(\theta+(1-\alpha)c)^2-(\alpha^2-\alpha)c^2\right)\right)dx\\
    &=\exp\left(\frac{(\alpha^2-\alpha)c^2}{2\sigma^2}\right)\frac{1}{\sqrt{2\pi}\sigma}\int_{-a}^a\exp\left(-\frac{\left(x-(\theta+(1-\alpha)c)\right)^2}{2\sigma^2}\right)\\
    &=\exp\left(\frac{(\alpha^2-\alpha)c^2}{2\sigma^2}\right)\Delta(\theta+(1-\alpha)c)
\end{align*}

Hence, we have
\begin{align*}
    &D_{\alpha}(\mathcal{N}^T(\theta,\sigma^2,\mathcal{B})\|\mathcal{N}^T(\theta+c,\sigma^2,\mathcal{B}))\\ 
    &=\frac{1}{\alpha-1}\log\Bigg(\left(\frac{\Delta(\theta+c)}{\Delta(\theta)}\right)^{\alpha-1}\left(\frac{\Delta(\theta+(1-\alpha)c)}{\Delta(\theta)}\right)\exp\left(\frac{(\alpha^2-\alpha)c^2}{2\sigma^2}\right)\Bigg)\\
    &=D_{\alpha}(\mathcal{N}(\theta,\sigma^2)\|\mathcal{N}(\theta+c,\sigma^2))+\log\frac{\Delta(\theta+c)}{\Delta(\theta)}+\frac{1}{\alpha-1}\log\frac{\Delta(\theta+(1-\alpha)c)}{\Delta(\theta)}
\end{align*}

Consider the case of rectified Gaussian:
    \begin{align*}
        &D_{\alpha}(\mathcal{N}^R(\theta,\sigma^2,\mathcal{B})\|\mathcal{N}^R(\theta+c,\sigma^2,\mathcal{B}))\\ &= \frac{1}{\alpha-1}\log\Bigg(\frac{\Phi\left(\frac{-a-\theta}{\sigma}\right)^\alpha}{\Phi\left(\frac{-a-\theta-c}{\sigma}\right)^{\alpha-1}}
    +\frac{\Phi\left(\frac{\theta-a}{\sigma}\right)^\alpha}{\Phi\left(\frac{\theta+c-a}{\sigma}\right)^{\alpha-1}}+\int_{-a}^a\frac{1}{\sigma}\phi\left(\frac{x-\theta}{\sigma}\right)^{\alpha}\phi\left(\frac{x-\theta-c}{\sigma}\right)^{1-\alpha}dx\Bigg)\\
    &=\frac{1}{\alpha-1}\log\Bigg(\frac{\Phi\left(\frac{-a-\theta}{\sigma}\right)^\alpha}{\Phi\left(\frac{-a-\theta-c}{\sigma}\right)^{\alpha-1}}
+\frac{\Phi\left(\frac{\theta-a}{\sigma}\right)^\alpha}{\Phi\left(\frac{\theta+c-a}{\sigma}\right)^{\alpha-1}}+\exp\left(\frac{(\alpha^2-\alpha)c^2}{2\sigma^2}\right)\Delta(\theta+(1-\alpha)c)\Bigg)
    \end{align*}

\end{proof}


\section{Proof of Lemma \ref{lemma:monotonicity}}
\label{append:dp_monotonicity}
\begin{proof}
It suffice to show that the derivative w.r.t $x$ is positive for $x>0$. For the rest of the proof, we assume $\sigma=1$ WLOG.
    \begin{align*}
        &\frac{\partial}{\partial x}D_{\alpha}\left(\mathcal{N}^T\left(\theta,\sigma^2,\mathcal{B}\right)\|\mathcal{N}^T\left(\theta+x,\sigma^2,\mathcal{B}\right)\right)\\ &= \alpha x+\frac{\partial}{\partial x}\log\left(\Phi(a-\theta-x)-\Phi(-a-\theta-x)\right)+\frac{1}{\alpha-1}\log \left(\Phi(a-\theta-(1-\alpha)x)-\Phi(-a-\theta-(1-\alpha)x)\right)\\
        &= \alpha x+\frac{-\phi(a-\theta-x)+\phi(-a-\theta-x)}{\Phi(a-\theta-x)-\Phi(-a-\theta-x)}+\frac{1}{\alpha-1}\frac{-(1-\alpha)\phi(a-\theta-(1-\alpha)x)+(1-\alpha)\phi(-a-\theta-(1-\alpha)x)}{\Phi(a-\theta-(1-\alpha)x)-\Phi(-a-\theta-(1-\alpha)x))}\\
        &= \alpha x+g(\theta+x)-g(\theta+(1-\alpha)x)
    \end{align*}
    where $g(x)=\frac{-\phi(a-x)+\phi(-a-x)}{\Phi(a-x)-\Phi(-a-x)}$.
    It suffice to show that $g(x-k)-g(x)\leq k$ for all $x$ and $k>0$. Now observe that
    \begin{align*}
        h(x)&=\frac{\partial g(x)}{\partial x}\\ &= -\frac{\partial -g(x)}{\partial x}\\
        &=-(-\mathcal{I}_{TN}(x)+1)\\
        &= \mathcal{I}_{TN}(x)-1\\
        &\geq -1
    \end{align*}
    where $\mathcal{I}_{TN}(x)$ is the fisher information at $x$. By Mean Value Theorem, for all $x$ and $k>0$, there exists $x'$ such that $x-k<x'<x$ and $h(x')=\frac{g(x)-g(x-k)}{k}$. Therefore, $\frac{g(x)-g(x-k)}{k}\geq -1\Rightarrow g(x-k)-g(x)\leq k$.
    
\end{proof}
\section{Proof of Theorem \ref{th:trunc_dp_amplification}}
\label{appen:dp_amp}
\begin{proof}
    Assume $\sigma=1$ here, WLOG. The RHS is simply $\alpha c^2$ because of the shift invariance of R\'enyi Divergence evaluated on two Gaussians. Based on the derivation in Lemma \ref{lemma:tgm_rdp}, it suffices to show that \[\log\frac{\Phi(a-\theta-c)-\Phi(-a-\theta-c)}{\Phi(a-\theta)-\Phi(-a-\theta)}+\frac{1}{\alpha-1}\log\frac{\Phi(a-\theta-(1-\alpha)c)-\Phi(-a-\theta-(1-\alpha)c))}{\Phi(a-\theta)-\Phi(-a-\theta)}\leq 0\]
    Let $f(x)=\log\left(\Phi(a-x)-\Phi(-a-x)\right)$ and we wish to prove that
    \[f(\theta+c)-f(\theta)+\frac{1}{\alpha-1}\left(f(\theta+(1-\alpha)c)-f(\theta)\right)\leq 0\]
    
    Note that by Lemma \ref{lemma:concave}, $f$ is concave. By Jensen's Inequality, we have
    \begin{align*}
        &\frac{f(\theta+c)+\frac{1}{\alpha-1}f(\theta+(1-\alpha)c)}{\frac{\alpha}{\alpha-1}}\leq f\left(\frac{\theta+c+\frac{1}{\alpha-1}(\theta+(1-\alpha)c)}{\frac{\alpha}{\alpha-1}}\right)\\
        \Leftrightarrow&f(\theta+c)+\frac{1}{\alpha-1}f(\theta+(1-\alpha)c)\leq \frac{\alpha}{\alpha-1}f(\theta)
    \end{align*}
    
\end{proof}
\section{Proof of Proposition \ref{prop:multidim}}
\begin{proof}
Note that with $L_{\infty}$ bounded support, both multivariate rectified and truncated Gaussian distribution is product distribution. In other words, the multivariate bounded Gaussian pdf $\mathbf{p}(\mathbf{x})$ could be written as the product of scalar bounded Gaussian pdf for each dimension $\prod_ip(\mathbf{x}_i)$. Given $\mathcal{B}=\{x|\|x\|_{\infty}\leq a\}$, we have
\begin{align*}
    D_{\alpha}\left(\mathbf{p}\|\mathbf{q}\right) &= \frac{1}{\alpha-1}\log\int_{\mathcal{B}}\mathbf{p}^\alpha\mathbf{q}^{1-\alpha}d\mathbf{x}\\
    &=\frac{1}{\alpha-1}\log\int_{S}\cdots\int_S\left(\prod_j\mathbf{p}_j\right)^\alpha\left(\prod_j\mathbf{q}_j\right)^{1-\alpha}dx_1\cdots dx_d\\
    &=\frac{1}{\alpha-1}\log\prod_j\int_S\mathbf{p}_j^\alpha\mathbf{q}_j^{1-\alpha}dx_j\\
    &=\sum_j\frac{1}{\alpha-1}\log\int_S\mathbf{p}_j^\alpha\mathbf{q}_j^{1-\alpha}dx_j\\
    &=\sum_jD_{\alpha}(p_j\|q_j)
\end{align*}

Now we compute the $FIL$ of the bounded Gaussian mechanism in the multi dimension case. Let $\mathbf{p}^B(\mathbf{x})$ be the probability density function of the multivariate bounded Gaussian mechanism. Hence, $\mathbf{p}^B$ could wither be the rectified Gaussian pdf $\mathbf{p}^R$ or the truncated Gaussian pdf $\mathbf{p}^T$. Note since each coordinate is independent of the others, we have
\begin{align*}
    \frac{\partial}{\partial\boldsymbol{\theta}}\mathbf{p}^B(\mathbf{x})=\left[\frac{\partial}{\partial\theta_1}p_1^B(\mathbf{x}_1),\cdots,\frac{\partial}{\partial\theta_d}p_d^B(\mathbf{x}_d)\right]
\end{align*}

\begin{align*}
    \frac{\partial^2}{\partial\boldsymbol{\theta}^2}\mathbf{p}^B(\mathbf{x})=\textbf{diag}\left(\left[\frac{\partial^2}{\partial\theta_1^2}p_1^B(\mathbf{x}_1),\cdots,\frac{\partial^2}{\partial\theta_d^2}p_d^B(\mathbf{x}_d)\right]\right)
\end{align*}

Hence, we have 
\begin{align*}
    \mathcal{I}_{\mathcal{M}}(D) &= \mathbb{E}\left[-\nabla_D^2\log \mathbf{p}^R(\mathbf{x}|D)\right]\\
    &= J_f^\top \mathbb{E}\left[\frac{\partial^2}{\partial\boldsymbol{\theta}^2}\log \mathbf{p}^B(\mathbf{x}) \right] J_f-\mathbf{H}\mathbb{E}\left[\frac{\partial}{\partial\boldsymbol{\theta}}\log \mathbf{p}^B(\mathbf{x}) \right]\\
    &= J_f^\top\textbf{diag}\left(\left[\boldsymbol{\eta}_1^2,\cdots,\boldsymbol{\eta}_d^2\right]\right) J_f
\end{align*}
where $\boldsymbol{\eta}_j:=\eta^B(\boldsymbol{\theta}_j,\sigma^2,S)$. 
    
\end{proof}
\section{Relation between quantization and rectified Gaussian}

As we mentioned earlier, our initial motivation comes from compression of SGD. In fact, we found that there's inherent relation between performing quantization over a Gaussian random variable and performing rectification over a Gaussian random variable. A quantization function takes a vector/scalar and an alphabet as input and dithers each scalar to its closest element on the alphabet. Then we have the following result:
\begin{theorem}
\label{th:quant}
    Let $q_k(\cdot)$ be a $k$-bit quantization function whose alphabet is within bounded range $[a,b]$ and $h$ be some hypothesis sampled from a Gaussian distribution: $h\sim\mathcal{N}(\theta,\sigma^2)$. We have the random variable $\lim_{k\rightarrow\infty}q_k(h)$ follows the distribution of $\mathcal{N}^R(\theta,\sigma^2,[a,b])$. Further, the same convergence result applies to the FIL as well.
\end{theorem}

\begin{proof}[Proof for Theorem \ref{th:quant}]
    Assume for now $\sigma=1$. Let the interval bounded with $[a,b]$ and the alphabets being $\left\{a,a+\frac{b-a}{k},\cdots,a+(k-1)\frac{b-a}{k},b\right\}$. Let $\Delta=\frac{b-a}{k}$ and $a_i=a+(i-1)\times\Delta$ for $i=1,\cdots,k+1$. We first calculate FIL of the $k$-bit quantization. Let $p_1=\Phi\left(\frac{(a_1+a_2)/2-\theta}{\sigma}\right)$, $p_i=\Phi\left(\frac{(a_i+a_{i+1})/2-\theta}{\sigma}\right)-\Phi\left(\frac{(a_{i-1}+a_{i})/2-\theta}{\sigma}\right)$ for $2\leq i\leq k$, and $p_{k+1} = \Phi\left(\frac{\theta-(a_{k}+a_{k+1})/2}{\sigma}\right)$. It's easy to calculate that
    \begin{equation*}
        \mathcal{I}_{q_k(X)}(\theta) = \sum_{i=1}^{k+1}\frac{1}{p_i} \left(\frac{\partial p_i}{\partial \theta}\right)^2
    \end{equation*}
Specifically, $\frac{\partial p_1}{\partial \theta}=-\frac{\phi\left(\frac{(a_1+a_2)/2-\theta}{\sigma}\right)}{\sigma}$, $\frac{\partial p_i}{\partial \theta}=\frac{\phi\left(\frac{(a_{i-1}+a_{i})/2-\theta}{\sigma}\right)-\phi\left(\frac{(a_{i}+a_{i+1})/2-\theta}{\sigma}\right)}{\sigma}$ for $2\leq i\leq k$, and $\frac{\partial p_k}{\partial \theta}=\frac{\phi\left(\frac{\theta-(a_{k}+a_{k+1})/2}{\sigma}\right)}{\sigma}$.
    
Now we are interested in computing the limit of the above quantity is equivalent to the FIL of rectified Gaussian:
\begin{align*}
    \lim_{k\rightarrow\infty}\mathcal{I}_{q_k(X)}(\theta) &= \lim_{k\rightarrow\infty}\sum_{i=1}^{k+1}\frac{1}{p_i} \left(\frac{\partial p_i}{\partial \theta}\right)^2\\
    &= \lim_{k\rightarrow\infty} \frac{\phi^2(a+\frac{b-a}{2k}-\theta)}{\Phi(a+\frac{b-a}{2k}-\theta)}+\frac{\phi(\theta-b+\frac{b-a}{2k})}{\Phi(\theta-b+\frac{b-a}{2k})}+\sum_{i=1}^{k}\frac{\left(\phi(a+\frac{(2i+1)(b-a)}{2k}-\theta)-\phi(a+\frac{(2i-1)(b-a)}{2k}-\theta)\right)^2}{\Phi(a+\frac{(2i+1)(b-a)}{2k}-\theta)-\Phi(a+\frac{(2i-1)(b-a)}{2k}-\theta)}\\
    &= \frac{\phi^2(a-\theta)}{\Phi(a-\theta)}+\frac{\phi^2(\theta-b)}{\Phi(\theta-b)}+\lim_{k\rightarrow\infty}\sum_{i=1}^{k}\frac{\left(\phi(a+\frac{(2i+1)(b-a)}{2k}-\theta)-\phi(a+\frac{(2i-1)(b-a)}{2k}-\theta)\right)^2}{\Phi(a+\frac{(2i+1)(b-a)}{2k}-\theta)-\Phi(a+\frac{(2i-1)(b-a)}{2k}-\theta)}\\
    &= \frac{\phi^2(a-\theta)}{\Phi(a-\theta)}+\frac{\phi^2(\theta-b)}{\Phi(\theta-b)}+\lim_{k\rightarrow\infty}\sum_{i=1}^{k}\frac{\left(\phi(c+\frac{2i+1}{2k}\Delta)-\phi(c+\frac{2i-1}{2k}\Delta)\right)^2}{\Phi(c+\frac{2i+1}{2k}\Delta)-\Phi(c+\frac{2i-1}{2k}\Delta)}
\end{align*}

It suffice to solve for that series $A=\lim_{k\rightarrow\infty}\sum_{i=1}^{k}\frac{\left(\phi(c+\frac{2i+1}{2k}\Delta)-\phi(c+\frac{2i-1}{2k}\Delta)\right)^2}{\Phi(c+\frac{2i+1}{2k}\Delta)-\Phi(c+\frac{2i-1}{2k}\Delta)}$.
\begin{align*}
    A &= \lim_{k\rightarrow\infty}\sum_{i=1}^{k}\frac{\frac{\left(\phi(c+\frac{2i+1}{2k}\Delta)-\phi(c+\frac{2i-1}{2k}\Delta)\right)^2}{\Delta/k}}{\frac{\Phi(c+\frac{2i+1}{2k}\Delta)-\Phi(c+\frac{2i-1}{2k}\Delta)}{\Delta/k}}\\
    &= \lim_{k\rightarrow\infty}\sum_{i=1}^{k}\frac{\left(\phi(c+\frac{2i+1}{2k}\Delta)-\phi(c+\frac{2i-1}{2k}\Delta)\right)\phi'(c+\frac{2i-1}{2k}\Delta)}{\phi(c+\frac{2i-1}{2k}\Delta)}\\
    &= \lim_{k\rightarrow\infty}\sum_{i=1}^{k}\frac{-(c+\frac{2i-1}{2k}\Delta)\left(\phi(c+\frac{2i+1}{2k}\Delta)-\phi(c+\frac{2i-1}{2k}\Delta)\right)\phi(c+\frac{2i-1}{2k}\Delta)}{\phi(c+\frac{2i-1}{2k}\Delta)}\\
    &= \lim_{k\rightarrow\infty}\sum_{i=1}^{k}-(c+\frac{2i-1}{2k}\Delta)\left(\phi(c+\frac{2i+1}{2k}\Delta)-\phi(c+\frac{2i-1}{2k}\Delta)\right)\\
    &= \lim_{k\rightarrow\infty}\left(c+\frac{1}{2k}\Delta\right)\phi\left(c+\frac{1}{2k}\Delta\right)+\sum_{i=1}^{k-1}\frac{\Delta}{k}\phi\left(c+\frac{2i+1}{2k}\Delta\right) -\left(c+\frac{2k-3}{2k}\Delta\right)\phi\left(c+\frac{2k-1}{2k}\Delta\right)\\
    &= (a-\theta)\phi(a-\theta)-(b-\theta)\phi(b-\theta)+\lim_{k\rightarrow\infty}\sum_{i=1}^{k-1}\frac{\Delta}{k}\phi\left(c+\frac{2i+1}{2k}\Delta\right)\\
    &= (a-\theta)\phi(a-\theta)-(b-\theta)\phi(b-\theta) + \Phi(b-\theta)-\Phi(a-\theta)
\end{align*}

Hence, \begin{align*}
    \lim_{k\rightarrow\infty}\mathcal{I}_{q_k(X)}(\theta)=\frac{\phi^2(a-\theta)}{\Phi(a-\theta)}+\frac{\phi^2(\theta-b)}{\Phi(\theta-b)} + \Phi(b-\theta) - \Phi(a-\theta) + (a-\theta)\phi(a-\theta) - (b-\theta)\phi(b-\theta)
\end{align*}
which is equivalent to $\eta^R$.
\end{proof}


\section{Hyperparameters}
\label{appen:hyperparam}
We list the details of hyperparameters we search for that produces the results in Figure \ref{fig:dp_utility_tradeoff} and Figure \ref{fig:fil_utility_tradeoff}.

\begin{table*}[h!]
	\centering
        
	\scalebox{0.87}{
	\begin{tabular}{ ll} 
	   \toprule[\heavyrulewidth]
	     
        \textbf{Hyperparameter} & Values \\
        \midrule
        $L_{\infty}$ clipping bound $C$ & $10^{-3},2\times 10^{-3},5\times 10^{-3},10^{-2}$\\
        Noise multiplier $\sigma/C$ & $10^3,2\times 10^3, 10^4, 2\times 10^4, 5\times 10^4$\\
        Bounded support parameter $a$ & $4,8,16,32,64$\\
        learning rate $lr$ & $8,16$\\
        \# of epochs $T$ & $10,20$\\
        Batch size (for FIL full training only) & $200$\\
        \toprule
       
    \bottomrule[\heavyrulewidth]
	\end{tabular}}
    \caption{CIFAR 10 Hyperparameters}
\end{table*}

\begin{table*}[h!]
	\centering
	\scalebox{0.87}{
	\begin{tabular}{ ll} 
	   \toprule[\heavyrulewidth]
	     
        \textbf{Hyperparameter} & Values \\
        \midrule
        $L_{\infty}$ clipping bound $C$ & $5\times 10^{-5}, 10^{-4},2\times 10^{-4},5\times10^{-4}$\\
        Noise multiplier $\sigma/C$ & $10^3,2\times 10^3, 5\times 10^3, 10^4, 2\times 10^4$\\
        Bounded support parameter $a$ & $2,4,8,16,32$\\
        learning rate $lr$ & $50,100$\\
        \# of epochs $T$ & $10,20$\\
        \toprule
       
    \bottomrule[\heavyrulewidth]
	\end{tabular}}
		\caption{CIFAR 100 Hyperparameters}
\end{table*}

\begin{table*}[h!]
	\centering
	\scalebox{0.87}{
	\begin{tabular}{ ll} 
	   \toprule[\heavyrulewidth]
	     
        \textbf{Hyperparameter} & Values \\
        \midrule
        $L_{\infty}$ clipping bound $C$ & $10^{-6}, 2\times 10^{-6},5\times10^{-6}$\\
        Noise multiplier $\sigma/C$ & $2\times 10^2,5\times 10^2, 10^3, 2\times 10^3, 5\times 10^3$\\
        Bounded support parameter $a$ & $0.01,0.02,0.05,0.1$\\
        learning rate $lr$ & $50,100,200$\\
        \# of epochs $T$ & $5,10,20$\\
        \toprule
       
    \bottomrule[\heavyrulewidth]
	\end{tabular}}
		\caption{OxfordIIITPet Hyperparameters}
\end{table*}

\end{document}